\documentclass[11pt]{article}

\newif\ifarxiv \arxivtrue

\usepackage[letterpaper,top=2cm,bottom=2cm,left=3cm,right=3cm,marginparwidth=1.75cm]{geometry}

\newcommand{\ARCjournal}{ARC Geophysical Research }
\newcommand{\ARCyear}{(2025) }
\newcommand{\ARCvolume}{1}
\newcommand{\ARCpaper}{---}

\usepackage{fancyhdr}
\ifarxiv
\else
    \fancypagestyle{plain}{\fancyhf{} \fancyhead[C]{{\scriptsize\color{blue} \ARCjournal \ARCyear \ARCvolume, \ARCpaper}} }
    \fancypagestyle{fancy}{\fancyhf{} \fancyhead[L]{{\scriptsize\it \ARCauthors}} \fancyhead[R]{{\scriptsize\it \ARCjournal \ARCyear \ARCvolume, \ARCpaper}}\fancyfoot[C]{\thepage} }
\fi

\usepackage{authblk}

\usepackage[numbers,sort&compress]{natbib}

\usepackage[colorlinks=true, allcolors=blue]{hyperref}
\usepackage{orcidlink}

\usepackage{amsmath}

\usepackage{graphicx}
\usepackage{caption}
\usepackage{subcaption}

\usepackage{xcolor}

\newcommand\blfootnote[1]{%
  \begingroup
  \renewcommand\thefootnote{}\footnote{#1}%
  \addtocounter{footnote}{-1}%
  \endgroup
}

\ifarxiv
\else
    \usepackage[
        type={CC},
        modifier={by-nc},
        version={4.0},
    ]{doclicense}
\fi

\usepackage{lineno}
\ifarxiv
\else
    \linenumbers
\fi

\usepackage{ctmmath-v3} 

\usepackage[utf8]{inputenc} 
\usepackage[T1]{fontenc}    
\usepackage{url}            
\usepackage{booktabs}       
\usepackage{amsfonts}       
\usepackage{nicefrac}       
\usepackage{microtype}      
\usepackage{algorithm}
\usepackage{algpseudocode}
\usepackage{array}
\usepackage{multirow}
\usepackage[english]{babel}
\usepackage{float}  
\newtheorem{proposition}[theorem]{Proposition}

\title{Calibrating Geophysical Predictions under Constrained Probabilistic Distributions}
\author[1]{Zhewen Hou\orcidlink{0009-0004-6506-2736}}
\author[3]{Jiajin Sun\orcidlink{0000-0000-0000-0000}}
\author[5]{Subashree Venkatasubramanian\orcidlink{0009-0007-9169-6814}}
\author[4]{Peter Jin\orcidlink{0009-0001-1171-4944}}
\author[2]{Shuolin Li\orcidlink{0000-0002-5719-1413}$^*$}
\author[1,2]{Tian Zheng\orcidlink{0000-0003-4889-0391}\thanks{Corresponding Author}}

\affil[1]{{\scriptsize Department of Statistics, Columbia University, NY}}
\affil[2]{{\scriptsize NSF Science and Technology Center for Learning the Earth with AI and Physics, Columbia University, NY}}
\affil[3]{{\scriptsize Department of Statistics, Florida State University, FL}}
\affil[4]{{\scriptsize Department of Applied Physics and Applied Mathematics, Columbia University, NY}}
\affil[5]{{\scriptsize Department of Computer Science, Columbia University, NY}}

\date{\today}
\newcommand{\ARCauthors}{Hou et al.}

\begin{document}

\ifarxiv
    \maketitle
\else
    \begin{figure}[h]
        \centering
        \includegraphics[width=\textwidth]{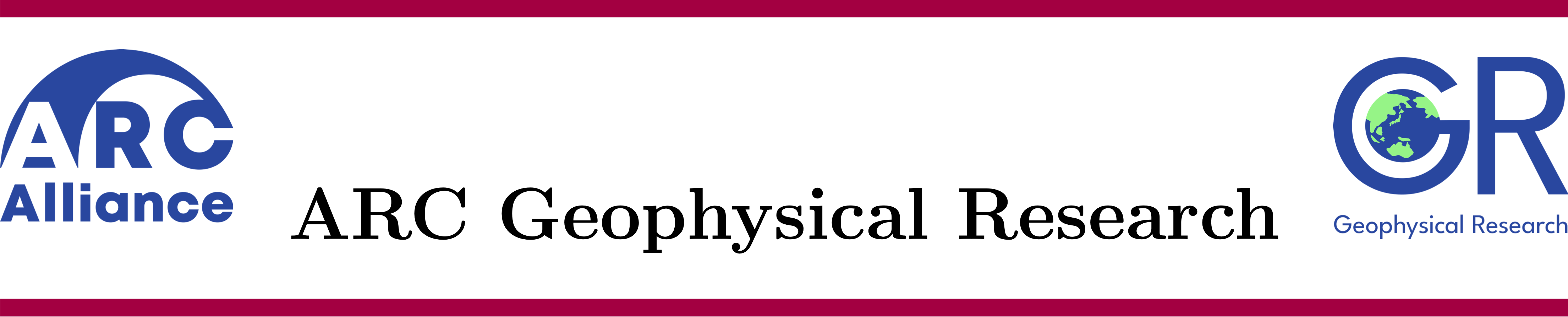}
    \end{figure}
    {\let\newpage\relax\maketitle}
    
    \noindent
    \hrule
\fi

\begin{abstract}
Machine learning (ML) has shown significant promise in studying complex geophysical dynamical systems, including turbulence and climate processes. Such systems often display sensitive dependence on initial conditions, reflected in positive \textit{Lyapunov exponents}, where even small perturbations in short-term forecasts can lead to large deviations in long-term outcomes. Thus, meaningful inference requires not only accurate short-term predictions, but also consistency with the system’s long-term attractor that is captured by the marginal distribution of state variables. Existing approaches attempt to address this challenge by incorporating spatial and temporal dependence, but these strategies become impractical when data are extremely sparse. In this work, we show that prior knowledge of marginal distributions offers valuable complementary information to short-term observations, motivating a distribution-informed learning framework. We introduce a calibration algorithm based on normalization and the \textit{Kernelized Stein Discrepancy} (KSD) to enhance ML predictions. The method here employs KSD within a \textit{reproducing kernel Hilbert space} to calibrate model outputs, improving their fidelity to known physical distributions. This not only sharpens pointwise predictions but also enforces consistency with non-local statistical structures rooted in physical principles. Through synthetic experiments—spanning offline climatological CO\textsubscript{2} fluxes and online quasi-geostrophic flow simulations—we demonstrate the robustness and broad utility of the proposed framework.
\vskip  8 pt
\noindent
\vskip 8 pt
\end{abstract}
\blfootnote{E-mail addresses: tian.zheng@columbia.edu; shuolin.li@columbia.edu}
\blfootnote{\vspace{-15pt}}
\vspace{-40 pt}

\newpage
\section{Introduction} \label{intro}
Dynamical systems are central to a wide range of mechanical and geophysical problems, including weather and climate modeling, fluid mechanics, and control systems. These systems are governed by differential equations that describe how state variables evolve over time. Accurate modeling of such evolution is crucial for forecasting, adaptation, and decision-making. However, many real-world systems are nonlinear, high-dimensional, and chaotic, making them extremely challenging to model or simulate with fidelity. Recent advances in machine learning (ML) have enabled the development of powerful surrogate models for dynamical systems, yielding gains in computational efficiency, expressiveness, and short-term predictive accuracy across various domains (e.g., \cite{kovachki, ross2023benchmarking, li2024machine}). Nonetheless, these data-driven models often face significant limitations in scientific settings. In particular, they tend to require dense and representative training data, which is rarely available in many applications. For instance, estimating air-sea CO\textsubscript{2} fluxes is critical for understanding the ocean’s role in regulating atmospheric carbon. Yet the observations are extremely sparse and biased, covering less than 1\% of the ocean’s surface and exhibiting substantial spatial sampling gaps \citep{socat1, socat2}.

Even when data are sufficient for training, surrogate models are typically optimized to minimize short-term prediction error (e.g., one-step forecast loss), often ignoring the long-term behavior of the system. This creates a fundamental mismatch in many physical systems governed by feedbacks and noise:
${dy}/{dt} = f(y) + \xi$,
where $f$ encodes the deterministic dynamics and $\xi$ captures stochasticity or unresolved processes. Models trained under such myopic loss functions may appear accurate in offline settings but fail during long-term rollout—drifting from the true attractor, violating conservation laws, or misrepresenting key statistical properties.

The quasi-geostrophic (QG) turbulence model offers a representative example. It is a widely used testbed in geophysical fluid dynamics due to its nonlinear interactions and rich turbulent behavior, yet reduced complexity \cite{kovachki, ross2023benchmarking}. ML surrogates trained on QG data often perform well over short horizons but diverge in online deployments, producing unphysical spectra or failing to capture the system’s invariant measure \cite{frezat2022posteriori, kovachki, ross2023benchmarking}. 
For example, in Figure~\ref{fig:pyqg}, a fully connected neural network (FCNN) \citep{ross2023benchmarking} is used as a surrogate parameterization for the turbulent flux term in a low-resolution QG model, trained to match turbulent quantities computed from a high-resolution reference simulation. The panels show four standard spectral diagnostics as functions of zonal wavenumber, computed a posteriori from long online integrations of each model. In all four diagnostics, the high-resolution model (blue solid line) exhibits a characteristic spectral shape and amplitude across scales that we regard as the target behavior. The coarse low-resolution model (orange dashed line) substantially underestimates the magnitude of energy transfers and misrepresents the overall level of generation and dissipation in the system. Augmenting the low-resolution model with the FCNN (green dashed line) yields only marginal changes to these spectra: in some panels there is a slight reduction of the bias relative to the purely low-resolution model, but overall the FCNN curves remain much closer to the low-resolution baseline than to the high-resolution reference. In particular, the FCNN does not recover the high-resolution levels of kinetic and available potential energy transfer or the correct balance between generation and friction. Thus, even though the FCNN improves short-term predictions of the resolved velocity fields, it fails to restore the correct long-term statistical structure of the QG turbulence, underscoring the need for methods that explicitly incorporate physical or statistical knowledge beyond one-step forecast losses.


\begin{figure}[htbp!]
\centering 
\vspace{-.3cm}
\includegraphics[width=0.8\textwidth]{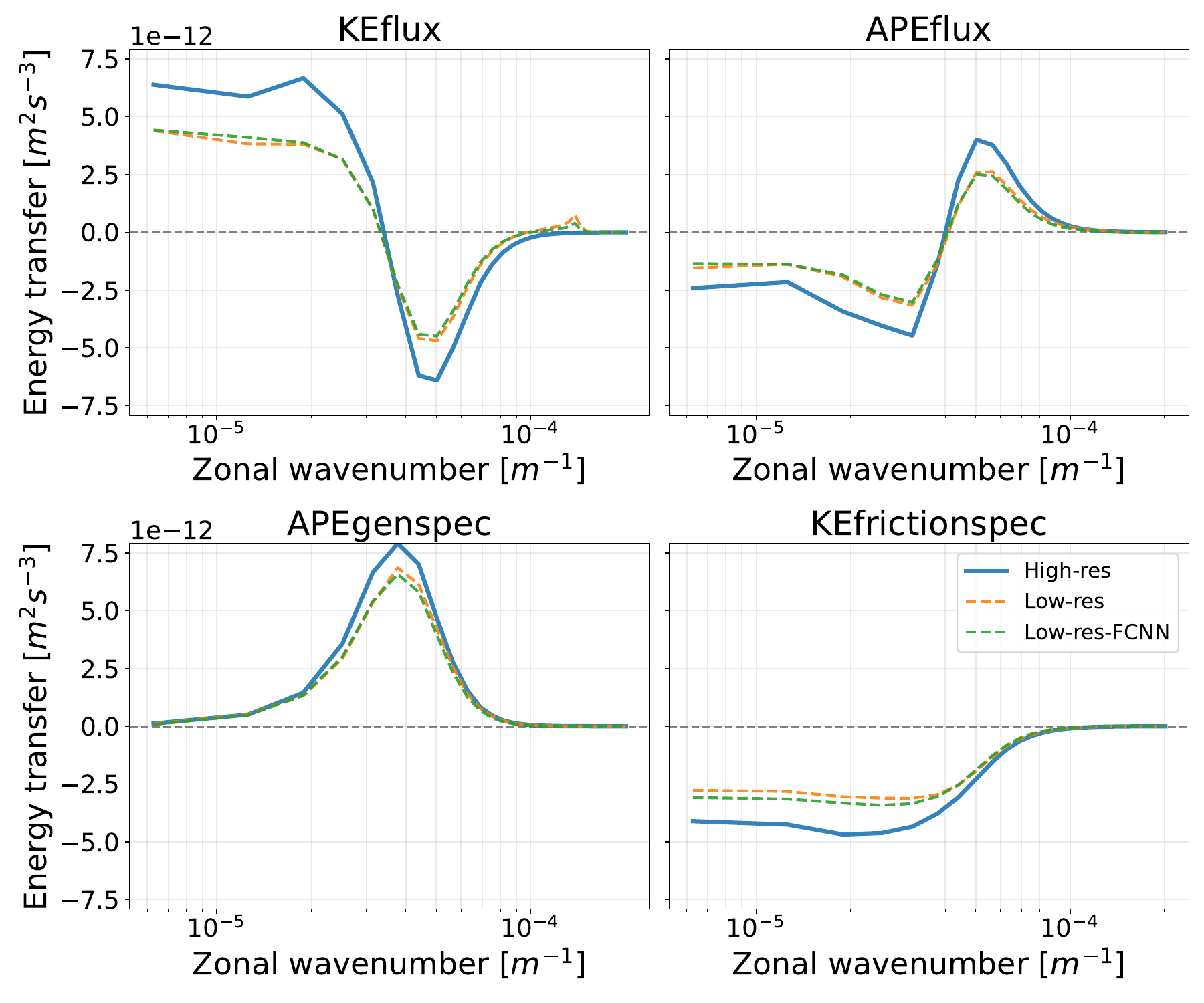} 
\caption{Spectral energy diagnostics for quasi-geostrophic turbulence. Each panel shows a derived statistical quantity (computed from the model trajectories, not directly learned by the FCNN) as a function of zonal wavenumber. The high-resolution simulation (blue solid line) provides the reference spectra. The coarse low-resolution model (orange dashed line) systematically underestimates the magnitude of energy transfers and misrepresents the level of generation and friction, with pronounced differences in several of the spectra. Adding an FCNN-based parameterization (green dashed line) produces only modest changes: across all four diagnostics, the FCNN curves remain much closer to the low-resolution baseline than to the high-resolution reference, with slight improvements in some wavenumber ranges but no recovery of the high-resolution spectral amplitudes. Panel titles indicate, respectively, spectra of kinetic energy flux (KEflux), available potential energy flux (APEflux), available potential energy generation (APEgenspec), and kinetic energy frictional effects (KEfrictionspec). Together, these results illustrate that a surrogate trained solely on short-horizon prediction errors can improve pointwise fields yet still fail to reproduce the correct long-term spectral and statistical structure of the turbulent flow.
}
\label{fig:pyqg}
\end{figure}

Physics-informed ML approaches such as PINNs \citep{cuomo2022scientific} and Bayesian neural networks \citep{jospin2022hands} attempt to embed physical structure directly into the learning process. However, these methods can be brittle under data sparsity or vanishing gradient, difficult to tune, or inapplicable when governing equations are unknown. Moreover, they do not fully exploit an important yet underutilized source of information: long-term steady-state distribution of the system. Dynamical systems often admit stationary distributions—statistical summaries that reflect the underlying physics and are robust to data irregularities. Prior work in stochastic modeling has shown that such distributions can encode essential features of the system \citep{gottwald2016stochastic}, yet they are rarely used directly to correct ML models' online behaviors.
In this paper, we propose a \textit{post hoc} correction framework that leverages knowledge of steady-state distribution of the system to improve the fidelity of pretrained ML surrogates. Rather than modifying the training objective or architecture, we introduce a lightweight, model-agnostic calibration method based on output normalization and kernelized Stein discrepancy (KSD). This approach adjusts model predictions to better align with the system’s physical statistics while preserving short-term accuracy.

The main contributions of this paper are threefold. First, we introduce a novel post hoc calibration method based on normalization and kernelized Stein discrepancy, which offers a principled alternative to hybrid loss approaches by leveraging knowledge of steady-state distribution of the system to enhance predictive accuracy. Second, we provide an intuitive understanding of the method by drawing conceptual connections to empirical Bayes techniques, supported by illustrative toy examples that clarify its underlying mechanism. Finally, we demonstrate the practical effectiveness of the proposed approach through applications to air-sea CO\textsubscript{2} flux estimation and turbulent flux modeling in quasi-geostrophic flows, showing improved performance and validating its utility in real-world scenarios.

\section{Background and Motivation}
\label{background}
\label{sec:ml4ds}

\subsection{Distribution-Informed Machine Learning for Dynamical Systems Projections} 

One natural approach to incorporate expected distributional behaviors of a dynamical system in ML model training is through a hybrid loss function during offline training that combines point-wise short-term prediction losses with regularization terms for the distributional differences between predictions and ground truth \citep {schiff2024dyslim, platt}. However, stochastic climate theory \citep{gottwald2016stochastic,franzke2020structure} suggests that the stable measure over attractors at a larger scale may depend more on system coefficients than only the instantaneous past or forcings at a smaller scale \citep{ghil}. It is therefore of concern that traditional ML emulators fitted using individual steps may primarily capture the forced response of the system and not the overarching distribution, due to the system’s chaotic nature. \citet{chattopadhyay_data-driven_2020} evaluated RC–ESN (reservoir computing-echo state network), ANN (a feed-forward artificial neural network), and RNN-LSTM (recurrent neural network with long short-term memory), for both short-term predictions and reproducing the long-term statistics. ANN could only predict short-term trends but failed to track the steady-state distribution. RC-ESN and RNN-LSTM managed to reproduce long-term probability distribution even after ``losing the trajectory.''  This provides empirical evidence that accurate short-term prediction and reproducing steady-state distribution draw information from data differently. 

\subsection{Kernelized Stein Discrepancy (KSD)}\label{sec:ksd-bkgrnd}

KSD \citep{gorham2015measuring} has become a highly versatile tool in ML. Originally proposed as a powerful tool for measuring the goodness-of-fit for probabilistic models, KSD has been extended to other ML tasks such as variational inference \citep{liu2016stein} and model selection \citep{gorham_measuring_2020}. KSD utilizes the properties of the Stein operator within a Reproducing Kernel Hilbert Space (RKHS) and offers a robust framework for testing the hypothesis that a sample is drawn from a specified probability distribution. For a density \( p \) and a differentiable function \( f \), the Stein operator \(\mathcal{A}_p\) is defined as:
\begin{equation}
    \mathcal{A}_p f(x) = s_p(x) f(x) + \nabla_x f(x),
\end{equation}
where $s_p(x) = \nabla_x \log p(x) = \frac{\nabla p(x)}{p(x)}$ is the (Stein) score function of $p$. The operator is constructed such that its expectation under the target distribution \(p\) is zero, i.e. $\mathbb E_{x \sim p} [\mathcal{A}_p f(x)] = 0$. Leveraging this property allows one to measure discrepancies \(\mathbb{E}_{x \sim q} [\mathcal{A}_p f(x)]\) between two distributions \(p\) and \(q\), which leads to the definition of the Stein discrepancy:
\begin{equation}
    \sqrt{S(p, q)} = \max_{f \in \mathcal{F}} \mathbb{E}_{x \sim q} [\mathcal{A}_p f(x)],
\end{equation}
where \(\mathcal{F}\) is an appropriate set of functions. To ensure \(\mathcal{F}\) is sufficiently broad, we can employ an infinite number of basis functions. This approach naturally leads to the KSD, where \(\mathcal{F}\) is a unit ball in a RKHS:
\begin{align}
     S(p, q) 
    = \mathbb{E}_{x, x' \sim p} [(s_q(x) - s_p(x)) k(x, x') (s_q(x') - s_p(x'))],
\end{align}
where \(k\) is a selected kernel function. This formulation allows the test to effectively capture complex dependencies within the data without requiring derivatives or integrals of the density \( p \). \citet{liu2016kernelized, chwialkowski2016kernel} propose adaptations of U-Statistics to estimate the Stein discrepancy:
\begin{equation}
    U_q(Z_{1:n}) = \frac{1}{n(n-1)} \sum_{1 \leq i \neq j \leq n} u_q (Z_i, Z_j),
\end{equation}
which is an unbiased estimator of \( S(p, q) \) where \( Z, Z' \) are independent copies from \( p \). Here, the alternative function \( u_q \) is typically defined in terms of the kernel function \( k \) and the score function $s_q$, given by:
\begin{align}
    u_q(x, x') = s_q(x) k(x, x') s_q(x') + s_q(x) \nabla_{x'} k(x, x') + s_q(x' )\nabla_x k(x, x') + \nabla_{x, x'} k(x, x'),
\end{align}
where \( \nabla_x k(x, x') \) and \( \nabla_{x'} k(x, x') \) are the gradients of the kernel function with respect to \( x \) and \( x' \), respectively. This function \( u_q \) captures the interaction between the gradients and the kernel, facilitating the computation of the Stein discrepancy.

%
%

\subsection{Related Work}

KSD is connected with other distributional measures in the literature. 
%
It can be seen as a generalization of Fisher Divergence under RBF kernel with bandwidth $h \to 0$ \citep{liu2016kernelized}. 
%
KSD can also be considered as a special case of MMD with a Steinalized kernel $u_q(x, x')$ that incorporates the score function $s_q$ of the target distribution $q$.
KSD is connected to the KL divergence as the derivative of KL divergence. \citet{liu2016stein} illustrated the relation in the context of variational inference. The KL divergence between an approximating distribution and the target distribution can be minimized by considering a small perturbation to the identity map, characterized by a smooth function \(\phi(x)\) and a small scalar \(\epsilon\). 
    Theorem 3.1 in \cite{liu2016stein} points out that
    \[
    \left. \nabla_{\epsilon} \mathrm{KL}(q_{[T]} \| p) \right|_{\epsilon=0} = -\mathbb{E}_{x \sim q} [ \text{trace}(\mathcal{A}_p \phi(x)) ],
    \]
    where \(\mathcal{A}_p \phi(x)\) is the Stein operator. This connection implies that the steepest descent direction in KL divergence corresponds to the optimal perturbation direction, which is precisely what KSD measures. 

\section{Distribution-Informed Prediction via KSD Calibration}
\label{sec:calibration}
Consider a learning task with training dataset $\mathcal{D}_{\text{train}}$, consisting of input features $\mathbf{X}_{\text{train}} \in \mathbb{R}^{d_x}$ and corresponding outputs $Y_\text{train} \in \mathbb{R}$, where the dependence between $\mathbf{X}$ and $Y$ is to be modeled by a ML model $M$. It is believed that a dynamical system generates $Y$, and prior knowledge suggests a marginal {\em knowledge probability distribution} $p$. 

\subsection{Calibration by KSD gradient descent}
 We propose a {\em post hoc} approach to improve predictions produced by conventional approaches mentioned in  {\bf Sec.}~\ref{sec:ml4ds}. Applying a fitted model $M(\cdot;\hat{\mathbf{\theta}})$ based on $\mathcal{D}_{\text{train}}$ to new input $\mathbf{X}$ provides us with the {\em raw} predictions $\hat{Y}$. Building on the KSD metric discussed in  {\bf Sec.}~\ref{sec:ksd-bkgrnd}, we propose a calibration update that aims to better align the raw predictions $\hat{Y}$ with the {\em knowledge distribution} $p$. This step does not require additional training data and operates solely on $\hat{Y}=\{\hat{y}_i, i=1, \ldots, n\}$ and $p$. 
%
KSD allows for a straightforward calculation of the discrepancy between a sample of observations and a knowledge distribution without the need for integration or optimization. It provides both a simple, intuitive interpretation and a solid theoretical foundation. In contrast, alternative statistical metrics would introduce additional computational overhead.  The gradient step can also be interpreted as a Wasserstein-2 (metric) gradient step using the KSD loss function, applied to the empirical distribution \(\frac{1}{n}\sum \delta_{y_i}\). 

We use the KSD $U_{p}(\hat{y}_{1:n})$ to measure how closely the predictions $\hat{y}_i$ align with $p$, where a smaller value of KSD indicates a better fit. Denote the derivatives of $U_p$ with respect to $\hat{y}_i$ as 
\[
D(\hat{y}_i) = \frac{\partial U_p(\hat{y}_{1:n})}{\partial \hat{y}_i}.\] 
We propose a gradient descent approach, selecting a small step size $\lambda$ to update the predictions as $\tilde{y}_i = \hat{y}_i - \lambda \cdot D(\hat{y}_i)$, optimizing $U_{p}$ to integrate information from both the raw predictions and the knowledge distribution.

\subsection{Theoretical Justification}

We consider a simple setting, where we assume the raw prediction equals the true outcome value plus some noise, i.e., $\hat y_i = y_i + \varepsilon_i $, $\varepsilon_i $ i.i.d.\ $\sim N(0,\sigma^2)$. We have the following result. The proof can be found in Appendix~\ref{theoryappendix}.

Let $\tilde y_i = \hat y_i - \lambda \frac{\partial U_p(\hat y_{1:n})}{\partial \hat y_i} $. Then $$\lambda ={ \sigma^2\sum_{i=1}^n\mathbb{E}\left\{\frac{\partial^2 U_p(\hat y_{1:n})}{\partial \hat y_i^2} \right\} } \bigg/ { \sum_{i=1}^n \mathbb{E}\left\{\frac{\partial U_p(\hat y_{1:n})}{\partial \hat y_i} \right\}^2 }$$ minimizes $\sum_{i=1}^n \mathbb E(\tilde y_i - y_i)^2 $. With this choice of $\lambda$, we have $\sum_{i=1}^n \mathbb E(\tilde y_i - y_i)^2 \leq \sum_{i=1}^n \mathbb E(\hat y_i - y_i)^2 $. This result aligns with the principle underlying the celebrated James-Stein estimator in Empirical Bayes, wherein the mean squared error (MSE) is reduced through a bias-variance trade-off. This establishes that, in most cases, there exists a nonzero optimal $\lambda$ that guarantees error (MSE) reduction.

\subsection{Algorithm Implementation with Early Stopping}

Guided by the properties of KSD, we construct the following calibration algorithm by applying KSD gradient descent updates iteratively. We set a small step size and a maximum cap for the number of updates. 
To avoid overfitting while having an automatic stopping rule for the KSD updates, we adopt the 2-Wasserstein loss (2-WD\textsuperscript{2}) between updated predictions and the knowledge distribution as the surrogate performance metric. Specifically, after each gradient step, we calculate the 2-Wasserstein loss. Training halts when the 2-Wasserstein loss shows no improvement for several steps (specified by a {\em patience} parameter).

{\em An optional normalization step}. In practice, the distribution of the raw predicted values could be far away from the knowledge distribution due to substantial estimation biases. The KSD algorithm still offers excellent performance gain but may stop too early. For such a scenario, we apply a \textit{normalizing} step to the raw predictions to align their distribution with the mean and standard deviation of the {\em knowledge distribution}. In addition, if the distribution information includes data boundaries, we apply clipping after normalization to ensure that the predicted values remain within these limits. 

Algorithm~\ref{alg:KSD} provides the pseudocode of a distribution-informed prediction workflow with the proposed KSD calibration. The {\em calibrated predictions} are denoted as $\tilde{y}_i$, $i=1,\ldots, n$. 

\begin{algorithm}[h]
\caption{Distribution-Informed Prediction via Kernelized Stein Discrepancy (KSD) Calibration.}
\label{alg:KSD}
\begin{algorithmic}[1]
\Require Training set $\mathcal{D}_{\text{train}}=(\mathbf{X}_{\text{train}}, Y_{\text{train}})$, evaluation input $\mathbf{X}_{\text{eval}}$, predictive model $M$, knowledge distribution $p$, step size $\lambda$, patience parameter $s$.
\Ensure Raw predictions $\hat{Y}$ and Calibrated predictions $\tilde{Y}$ for $\mathbf{X}$.
\State Train predictive model $M$ using $\mathcal{D}_{\text{train}}$.
\State Compute raw predictions: $\hat{Y} = M(\mathbf{X}_{\text{eval}}; \hat\theta)$.
\State Normalize $\hat{Y}$ to match the mean, variance, and range of $p$, yielding $\tilde{y}_i$.
\State \textbf{Repeat:}
    \begin{enumerate}
        \item Calculate the KSD gradient,  $
        D(\tilde{y}_i) = \frac{\partial U_p(\tilde{y}_{1:n})}{\partial \tilde{y}_i}$, for each normalized prediction using automatic differentiation.
        \item Update each prediction:
        \[
        \tilde{y}_i \leftarrow \tilde{y}_i - \lambda \cdot D(\tilde{y}_i).
        \]
        \item Compute the 2-Wasserstein loss between $p$ and the updated predictions: $\hat{W}_2(p, \{\tilde{y}_i\}_{i=1}^n)$.
        \item Stop if $\hat{W}_2(p, \{\tilde{y}_i\}_{i=1}^n)$ does not improve for $s$ consecutive steps.
    \end{enumerate}

\State \textbf{return} $\hat{Y} = \{\hat{y}_{i}, i=1, \ldots, n\}$, and $\tilde{Y} = \{\tilde{y}_{i}, i=1, \ldots, n\}$.

\end{algorithmic}
\end{algorithm}

\section{Experiments}
In this section, we illustrate the proposed calibration approach using one toy example, an application to the estimation of air-sea CO$_2$ fluxes using an ML-based reconstruction from surface ocean observations of the partial pressure of CO$_2$ (pCO\textsubscript{2}) \citep{gloege2021quantifying}, and an application to an online hybrid emulator of quasi-geostrophic turbulence \citep{ross2023benchmarking}.

\subsection{Experiment Setup}\label{sec:setup}

For our numerical experiments, we implement Algorithm~\ref{alg:KSD}.
The predictive model is trained on training data using the MSE loss function. As previously mentioned, the KSD calibration does not require training but uses a knowledge distribution. We evaluate the workflow on an independent evaluation dataset, comparing raw and calibrated predictions against the ground truth $Y$ values using both point-wise loss, evaluated by mean square error (MSE), and distribution loss, evaluated by 2-Wasserstein Loss (2-WD\textsuperscript{2}). As discussed previously, rather than creating a trade-off between point-wise prediction accuracy and distributional alignment, the KSD calibration is expected to improves both MSE and 2-WD\textsuperscript{2}. 

\subsubsection{Choice of kernel}

For all our experiments, we use the radial basis function (RBF) kernel for calculating KSD. \citet{liu2016kernelized} demonstrated that if the kernel belongs to the Stein class of distribution $p$, 
    \[
        S(p, q) = \mathbb{E}_{x, x' \sim p} [u_q(x, x')],
    \]
and the RBF kernel is in the Stein class for smooth densities supported on $\mathbb R$. 
We use a bandwidth of three standard deviations of the data. Additional experiments (not shown) were conducted using different bandwidth values and obtained similar results. 

\subsubsection{Step size and patience parameter}

The magnitude of the gradient with respect to KSD varies in practice. In our experiments, we scale the gradients so that their second moment matches the variance of the predictions. This allows us to regulate how much each gradient descent step impacts the predictions. After this scaling step, we choose $\lambda=0.01$ so that each step leads to a small update. We set the patience parameter to $s=20$, allowing sufficient steps to confirm the absence of further improvements. Numerical experiments show that our results are not sensitive to the choice of $s$.

\subsection{Toy Example: Shifted Linear Dynamical Systems}

We consider a toy linear dynamical system:
\[
\mathbf{X}_t = A \mathbf{X}_{t-1} + \boldsymbol{\varepsilon}_t
\]
where $\mathbf{X}_t \in \mathbb{R}^6$ is the state vector at time step $t$, $A \in \mathbb{R}^{6 \times 6}$ is the transition matrix that governs the linear dynamics of the system, and $\boldsymbol{\varepsilon}_t \sim \mathcal{N}(0, \Sigma)$ represents Gaussian noise at each time step, with zero mean and covariance matrix $\Sigma$. The matrix $A$ is sampled from a standard Gaussian distribution and scaled such that its largest eigenvalue $\lambda_{\max}$ is $0.99$, resulting in a system that is nearly chaotic yet stable.

To simulate distributional shifts, we generate training and testing datasets using the same transition matrix $A$ but different noise covariance matrices $\Sigma$, leading to distinct steady-state distributions. Specifically, the noise covariance $\Sigma$ is constructed by first generating a lower triangular matrix $L$ with entries drawn from a normal distribution with mean $0$ and standard deviation $100$, then computing $\Sigma = L L^T$ to ensure positive definiteness. This procedure is carried out independently for both the training and testing datasets, resulting in distinct covariance structures for each.

\subsubsection{Dataset Generation}

We generate datasets with sizes ranging from 100 to 1,000. To construct each dataset, we first initialize the system at $\mathbf{X}_0$ drawn from its stationary distribution. The sequence was simulated forward and subsampled using fixed intervals—\texttt{interval}, set to 10 unless otherwise specified. This provides a realistic setting where training and testing data differ not in local dynamics but in their marginal distributions.

To initialize the system, the initial state \(\mathbf{X}_0\) is drawn from the steady-state distribution. After simulating a long sequence of state transitions, every \texttt{interval} steps are selected to form the training dataset or the test dataset, ensuring the samples are sufficiently spaced to avoid redundancy.

\subsubsection{Machine Learning Model and Knowledge Distribution}

The goal is to predict the norm (length) of the future state vector, defined as \( Y_t = \|\mathbf{X}_{t+1}\|_2 / 100 \). To predict \(Y_t\) based on \(\mathbf{X}_t\), we employed LightGBM \cite{ke2017lightgbm} as the machine learning model. The model’s hyperparameters were automatically tuned using the Optuna framework. The training data was randomly split, with 80\% used for training and 20\% for validation, and the validation set guided the selection of the hyperparameter.

To estimate knowledge marginal distribution, we sampled 10,000 state vectors \(\mathbf{X}\) from the steady-state distribution of the testing dataset, propagate them one step, and collect the corresponding $Y$ values. These 10,000 \(Y\) values were then used to fit a Gaussian Mixture Model (GMM) with 5 components, which provided an estimate of the marginal distribution of \(Y\) in the testing dataset.

\subsubsection{Distributional Shift}

The primary challenge in this experiment stemmed from the differences in steady-state distributions between the training and testing datasets, caused by the distinct covariance matrices. This mismatch led to a systematic bias in the raw predictions of the norm of \(\mathbf{X}_{t+1}\), as illustrated in Figure~\ref{fig:dynamic}. The bias manifested as a distributional shift between the predicted and true values of the testing dataset.

\begin{figure}[htbp!]
    \centering
    \begin{tabular}{cc}
        {\small (A) Biased Predictions} & {\small (B) Distribution Calibration} \\
        \includegraphics[width=0.4\linewidth]{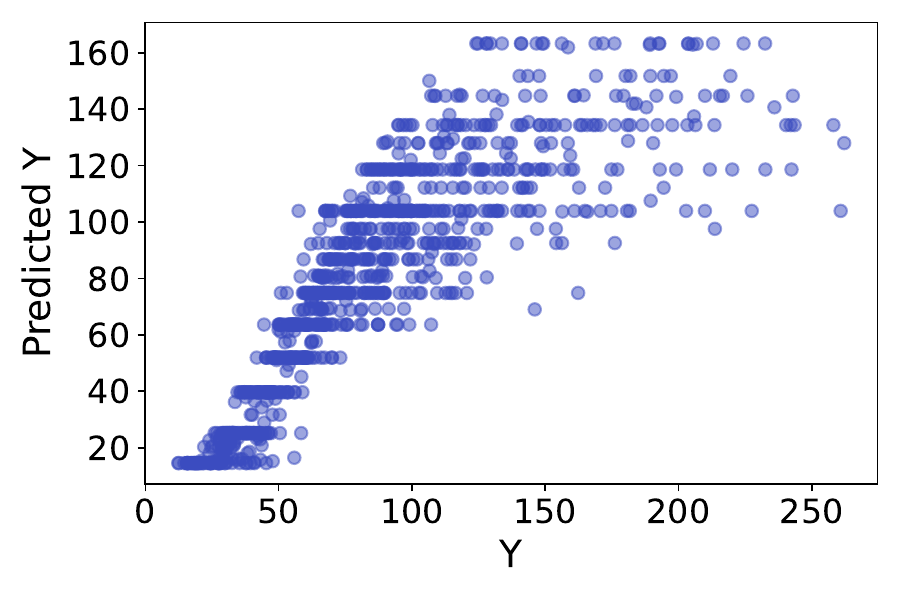}&
        \includegraphics[width=0.4\linewidth]{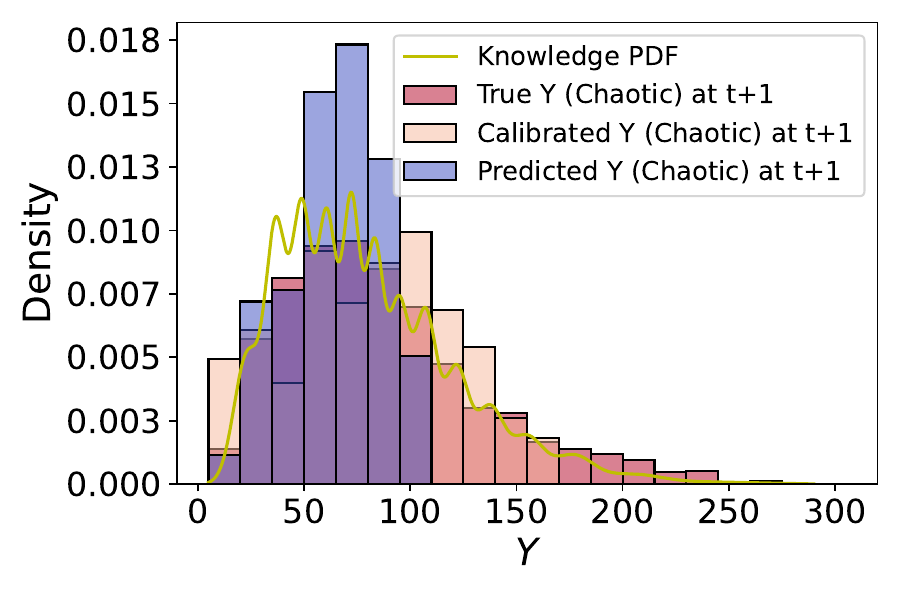} \\
        \multicolumn{2}{c}{
            (C) {\small MSE Reduction}
        } \\
        \multicolumn{2}{c}{
            \includegraphics[width=0.66\linewidth]{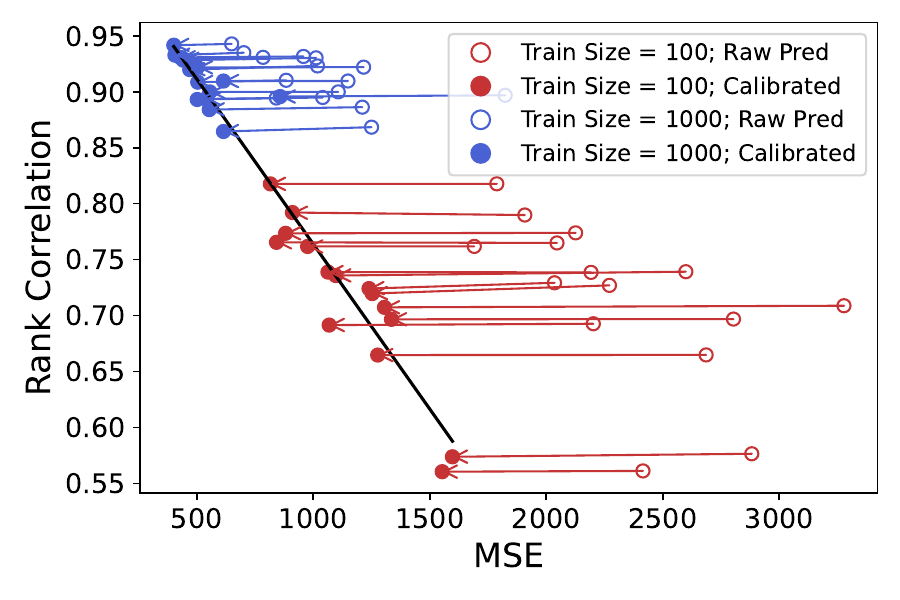}
        }
    \end{tabular}
    \vspace{-0.1cm}
    \caption{\textbf{KSD calibration under steady-state distribution shift.}
    The dynamical system is trained and evaluated under different steady-state distributions, inducing systematic bias in the raw predictions.
    (A) True versus raw predicted values of \(Y\) on the test set: the model captures the local dependence between \(\mathbf{X}_t\) and \(Y\), but the point cloud deviates from the identity line, revealing systematic bias driven by the distributional shift. 
    (B) One-dimensional marginal densities for the knowledge distribution, the true test responses, the raw predictions, and the calibrated predictions. The raw predictive distribution is shifted relative to both the knowledge and true test distributions, whereas KSD calibration transports it closer to these targets, improving distributional alignment. 
    (C) For each independent run, we plot the MSE against the rank correlation between predictions and ground truth. Hollow markers denote raw predictions and filled markers denote calibrated predictions from the same run; red (blue) points use a training size of 100 (1000). Calibration preserves rank correlation but consistently reduces MSE, and the linear fit over all calibrated points shows that post-calibration accuracy improves approximately linearly with the model’s rank correlation.}
    \vspace{-0.1cm}
    \label{fig:dynamic}
\end{figure}

Figure~\ref{fig:dynamic}-(A) plots the true values of \(Y\) against the corresponding raw predictions on the test set. The points concentrate along an increasing curve, indicating that the machine learning model successfully captures the short-term local dependence between \(\mathbf{X}_t\) and \(Y\). However, the cloud of points systematically deviates from the identity line: for example, large true values of \(Y\) tend to be under-estimated, reflecting a structural bias induced by the mismatch between the training and test steady-state distributions rather than pure noise in the regression. This panel therefore highlights that the model is predictive but is biased under distributional shift.

Figure~\ref{fig:dynamic}-(B) focuses on the one-dimensional marginal distributions. It overlays the knowledge distribution, the empirical distribution of the true test responses, and the distributions of the raw and calibrated predictions. The raw predictive distribution is compressed toward the left and carries almost no mass in the right tail, failing to cover a substantial portion of the support where the true and knowledge distributions place non-negligible probability. In other words, the raw model rarely generates large values of \(Y\) even though such values occur frequently under the true dynamics. After applying KSD calibration, the distribution of the calibrated predictions extends into this missing region and moves much closer to the knowledge and true test distributions, yielding substantially better alignment of the predictive marginal.


Each point in Figure~\ref{fig:dynamic}-(C) corresponds to one independent experimental run, where the training and test data are generated with a different random seed and a separate LightGBM model is trained. For each run, the horizontal axis reports the MSE and the vertical axis reports the rank correlation between the model predictions and the ground truth. Hollow markers show the performance of the raw LightGBM predictions, while the corresponding filled markers at the same vertical level show the performance after applying KSD-based calibration to that specific model: because the calibration only adjusts the predictive distribution without changing the relative ordering of predictions, the rank correlation is preserved and any gain appears purely as a horizontal shift to the left (lower MSE). The red points are obtained with a training set of size 100 and the blue points with a training set of size 1000, indicating that larger training sets yield base models with both lower raw MSE and higher rank correlation. Nonetheless, even these stronger base models still benefit from calibration, as seen from the consistent leftward shifts of the filled points. The black line is a linear fit over all filled markers, showing that the achievable post-calibration accuracy is approximately a linear function of the underlying model’s rank correlation: models with better ranking structure systematically obtain larger improvements from our calibration procedure.

\subsubsection{Results}

In this toy example, models trained on biased data have poor MSE but moderate rank correlations with the ground truth on the test data. This suggests that the predictions were {\em systematically} biased and could be calibrated to improve MSE. In such contexts, model training using a hybrid loss (e.g., MSE+Distribution Loss) with the training data would fail due to the disconnect between the marginal outcome distribution of the training data and the knowledge distribution. Our proposed KSD calibration offers an efficient alternative solution.

The results summarized in Table~\ref{tab:toy_example2_appendix} demonstrate a clear improvement in both prediction accuracy and distribution alignment due to the proposed calibration method. Overall, the MSE decreases by 40\%-50\%, indicating a significant enhancement in point-wise prediction accuracy. Furthermore, the 2-WD\textsuperscript{2}, which quantifies the difference in distributions between the predicted and actual lengths of \(\mathbf{X}_{t+1}\), is reduced to between one-third and one-eleventh of its original value.

\begin{table}[h]
\caption{\textbf{Effectiveness of KSD calibration under steady-state distribution shift.} 
MSE and 2-WD\textsuperscript{2} for different combinations of training and evaluation dataset sizes, comparing raw and calibrated predictions when the training and testing datasets are generated from different covariance matrices so their stationary distributions differ. Here, \textit{raw} refers to the predictive model's original output, and \textit{cali w/ norm} indicates calibration based on normalized predictions. Across all $(n_{\text{train}}, n_{\text{test}})$ pairs, KSD calibration substantially reduces both MSE and 2-WD\textsuperscript{2}, illustrating its ability to correct biases induced by steady-state distribution shifts. The mean and standard deviation of each metric are computed from $1000$ experiments.}
%
\label{tab:toy_example2_appendix}
\centering
\begin{tabular}{cc cc cc}
\hline
\multirow{2}{*}{$\mathbf{n_{train}}$} & \multirow{2}{*}{$\mathbf{n_{test}}$} & \multicolumn{2}{c}{\textbf{raw}} & \multicolumn{2}{c}{\textbf{cali w/ norm}} \\
 & & \textbf{MSE} & \textbf{2-WD\textsuperscript{2}} & \textbf{MSE} & \textbf{2-WD\textsuperscript{2}} \\
\hline
\multirow{2}{*}{100} & 100 & 2191.3 (1092.5) & 1832.3 (1014.3) & 1207.2 (487.7) & 292.8 (225.7) \\
& 1000 & 2257.8 (592.8) & 1916.9 (588.4) & 1160.2 (308.2) & 177.6 (93.4) \\
\hline
\multirow{2}{*}{1000} & 100 & 1035.3 (630.9) & 811.0 (558.1) & 614.2 (284.7) & 259.6 (193.3) \\
& 1000 & 1070.8 (279.0) & 861.2 (264.0) & 534.3 (103.5) & 158.5 (61.9) \\
\hline
\end{tabular}
\end{table}

We next isolate the role of steady-state distribution shift in driving the gains from KSD calibration. Table~\ref{tab:toy_example2_appendix} reports results for the toy example when the training and testing datasets are generated from different covariance matrices, so that their stationary distributions differ. In all $(n_{\text{train}}, n_{\text{test}})$ combinations, the calibrated predictions achieve substantially lower MSE and 2-WD\textsuperscript{2} than the raw model outputs, indicating that KSD calibration effectively corrects the bias induced by the mismatch between the training and test steady-state distributions.

To test whether these gains persist in the absence of distributional shift, Table~\ref{tab:toy_example2_appendix2} repeats the experiments under the same settings but with $\Sigma_{\text{train}} = \Sigma_{\text{test}}$, so that the training and testing datasets share the same stationary distribution. In this well-specified setting, calibration does not systematically improve performance: the MSE is often worse after calibration, and the reductions in 2-WD\textsuperscript{2} are small or appear only in a single configuration. Taken together, Table~\ref{tab:toy_example2_appendix} and Table~\ref{tab:toy_example2_appendix2} show that KSD calibration is most beneficial when there is a steady-state distribution shift to be corrected, rather than as a generic procedure that always shrinks prediction error.


\begin{table}[H]
\caption{\textbf{KSD calibration when there is no steady-state distribution shift.}
MSE and 2-WD\textsuperscript{2} for different combinations of training and evaluation dataset sizes when the training and testing datasets share the same covariance matrix, so that the stationary distribution remains consistent between training and testing. Here, \textit{raw} refers to the model's original output, and \textit{cali w/ norm} indicates calibration based on normalized predictions. In this well-specified setting, KSD calibration does \emph{not} consistently improve performance: MSE is often larger after calibration and 2-WD\textsuperscript{2} improves only marginally or in a single configuration. The mean and standard deviation of each metric are computed from $1000$ experiments. Together with Table~\ref{tab:toy_example2_appendix}, this highlights that the main benefit of KSD calibration arises when correcting steady-state distribution shifts, rather than in settings without such shifts.}
%
\label{tab:toy_example2_appendix2}
\centering
\begin{tabular}{cc cc cc}
\hline
\multirow{2}{*}{$\mathbf{n_{train}}$} & \multirow{2}{*}{$\mathbf{n_{test}}$} & \multicolumn{2}{c}{\textbf{raw}} & \multicolumn{2}{c}{\textbf{cali w/ norm}} \\
 & & \textbf{MSE} & \textbf{2-WD\textsuperscript{2}} & \textbf{MSE} & \textbf{2-WD\textsuperscript{2}} \\
\hline
\multirow{2}{*}{100} & 100 & 82.11 (52.96) & 30.13 (35.96) & 101.54 (53.33) & 33.07 (34.57) \\
& 1000 & 83.83 (31.33) & 26.93 (21.90) & 83.59 (27.62) & 9.64 (6.54) \\
\hline
\multirow{2}{*}{1000} & 100 & 16.26 (7.72) & 3.50 (4.20) & 42.46 (35.80) & 27.83 (33.55) \\
& 1000 & 17.08 (4.14) & 2.00 (1.95) & 20.17 (5.65) & 4.34 (4.36) \\
\hline
\end{tabular}
\end{table}

\subsection{Application to air-sea CO\textsubscript{2} flux}

\subsubsection{Background and Motivation}

We evaluate the proposed calibration's skill in improving the reconstruction of air-sea CO\textsubscript{2} fluxes based on scarce (1\% coverage), highly biased sample data \citep{socat1,socat2}. See the Appendix~\ref{pco2appendix} for a full discussion on these data issues. This level of data scarcity and sampling biases makes it infeasible to consider the hybrid loss, as discussed in {\bf Sec.}~\ref{sec:ml4ds}.

Understanding air-sea CO\textsubscript{2} fluxes is crucial for quantifying the ocean's role in the global carbon cycles.
Predicting air-sea CO\textsubscript{2} fluxes presents significant challenges due to the complex interplay of fphysical, chemical, and biological processes that govern the exchange of carbon dioxide between the atmosphere and the ocean. Factors such as ocean circulation, temperature, salinity, wind patterns, and biological activity all influence the fluxes, making accurate predictions difficult. Additionally, the spatial and temporal variability of these factors requires high-resolution data and sophisticated modeling techniques to capture the dynamics accurately.

\citet{landschutzer2016decadal} introduced a widely used SOM-FFN (Self-Organizing Map-Feed-Forward Network) to estimate the continuous monthly mean pCO\textsubscript{2}. \citet{gloege2021quantifying} used multiple Large Ensemble Earth system models \citep{gloege_2019} to evaluate the prediction ability of SOM-FFN. The Earth system models provided reanalysis data, allowing global coverage of the ground truth of pCO\textsubscript{2}. Applying the same sampling process as in the real observation data, the authors recreated the same level of data scarcity and sampling biases in the training set. They found that areas that lacked observation data had larger prediction biases/errors.

\subsubsection{Experiment Setup}

Following the setup of \citep{gloege2021quantifying}, we apply the proposed KSD calibration to SOM-FFN predictions using reanalysis data from the Community Earth System Model–Large Ensemble (RCP8.5) \citep{kay2015community}, where the global ground truth of pCO\textsubscript{2} is known. 
For our experiments, we treat sampled data from one ensemble member (CESM001) as ``real Earth'' data and focus on the 6th biogeochemical province identified by the SOM, which consists of disjoint regions in both the Northern and Southern Hemispheres and is characterized by highly heterogeneous and spatially structured data availability (Figure~\ref{fig:obv_pro15_com}).

\subsubsection{Knowledge Distribution}

\begin{figure}[h!]
    \centering
    \begin{tabular}{ccc}
        \subcaptionbox{January\label{fig:jan_density}}{\includegraphics[width=0.31\textwidth]{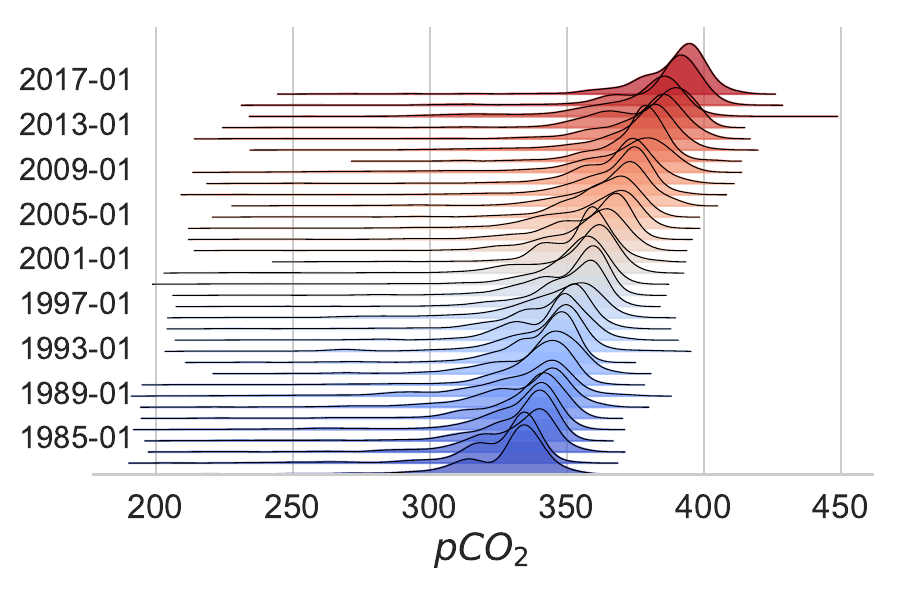}} &
        \subcaptionbox{February\label{fig:feb_density}}{\includegraphics[width=0.31\textwidth]{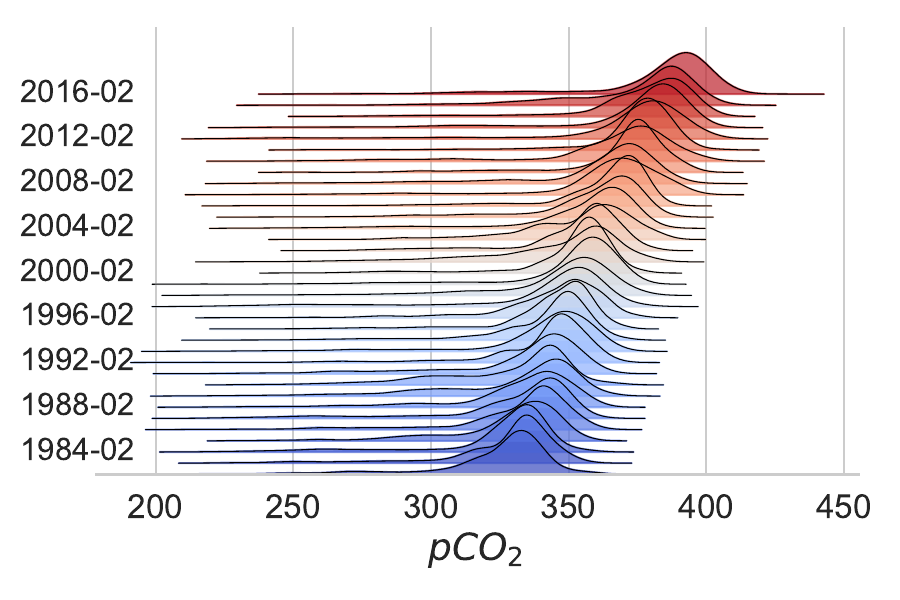}} &
        \subcaptionbox{March\label{fig:mar_density}}{\includegraphics[width=0.31\textwidth]{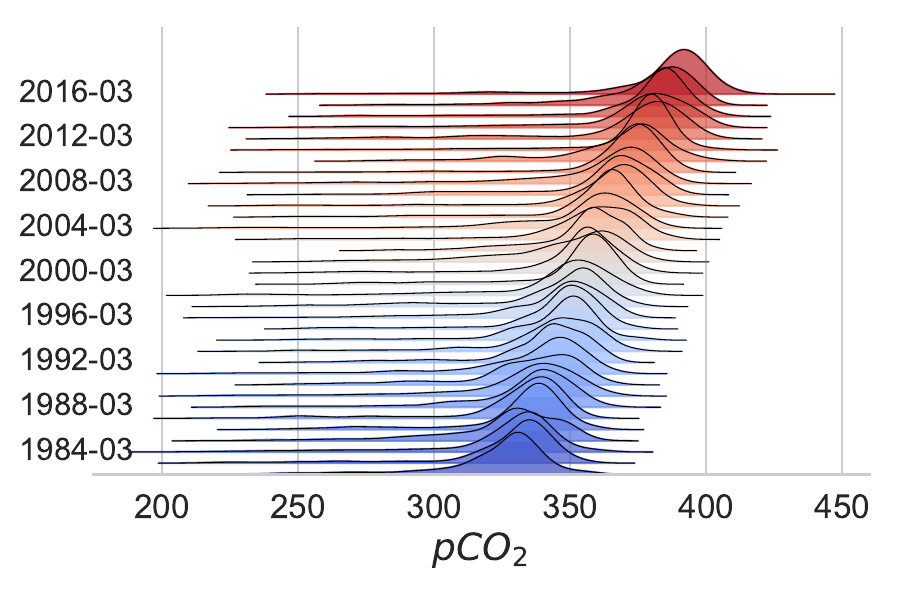}} \\
        \subcaptionbox{April\label{fig:apr_density}}{\includegraphics[width=0.31\textwidth]{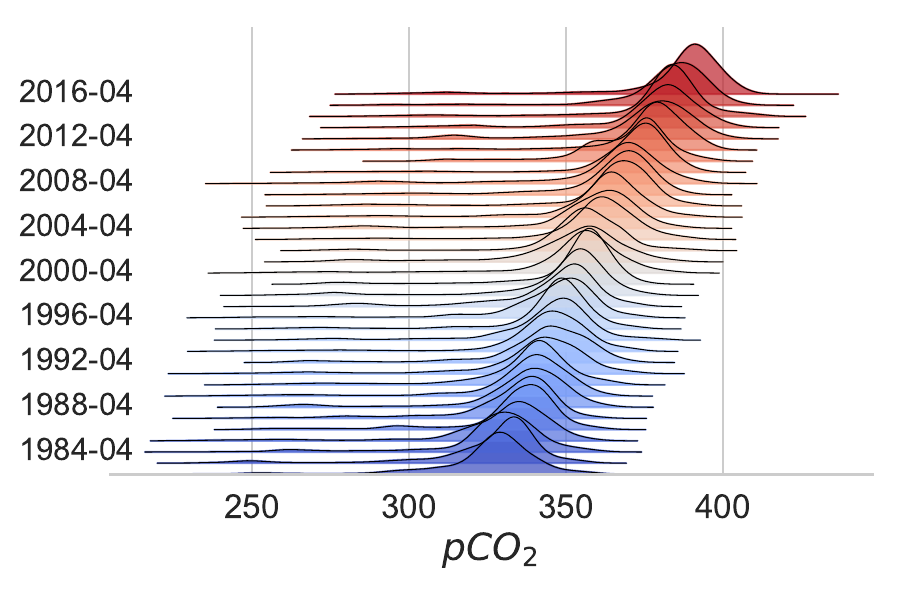}} &
        \subcaptionbox{May\label{fig:may_density}}{\includegraphics[width=0.31\textwidth]{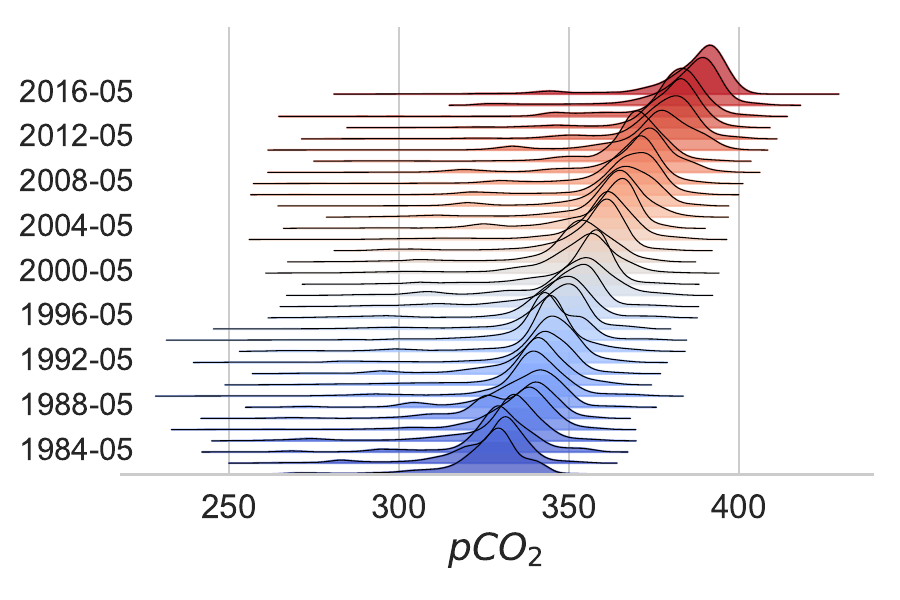}} &
        \subcaptionbox{June\label{fig:jun_density}}{\includegraphics[width=0.31\textwidth]{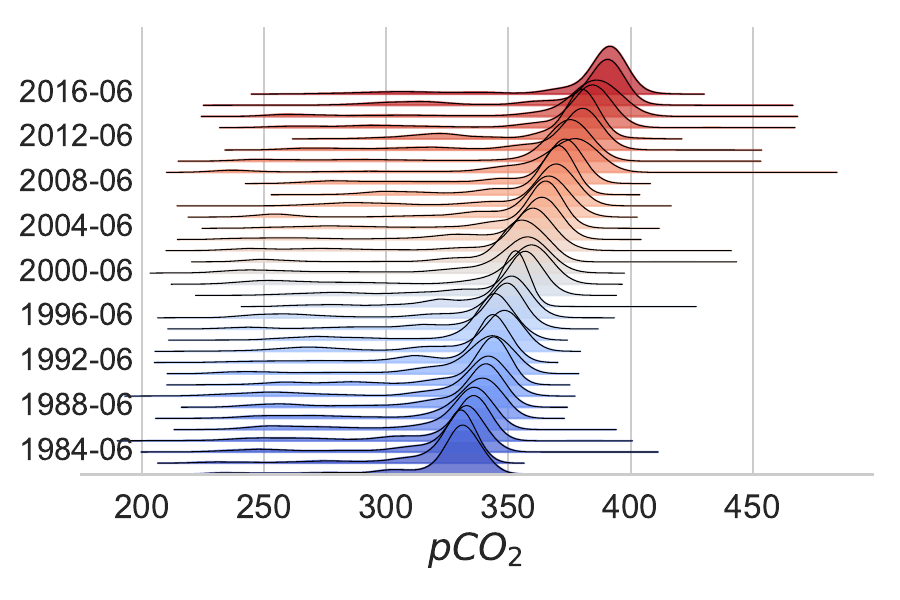}} \\
        \subcaptionbox{July\label{fig:jul_density}}{\includegraphics[width=0.31\textwidth]{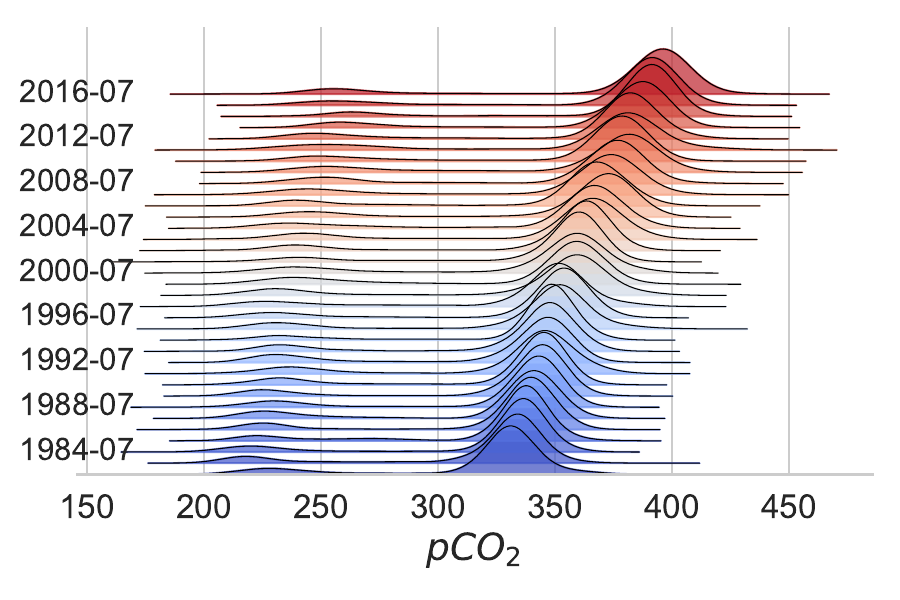}} &
        \subcaptionbox{August\label{fig:aug_density}}{\includegraphics[width=0.31\textwidth]{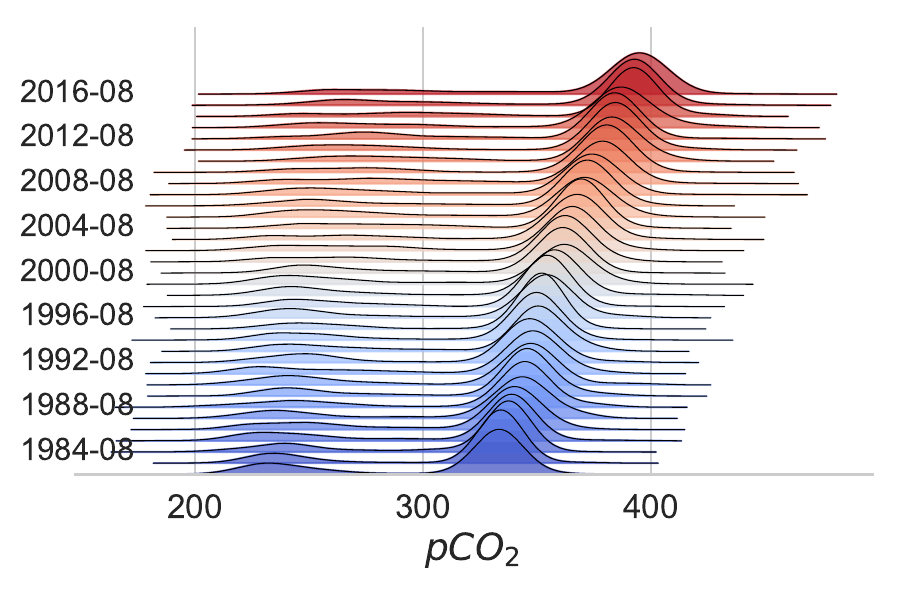}} &
        \subcaptionbox{September\label{fig:sep_density}}{\includegraphics[width=0.31\textwidth]{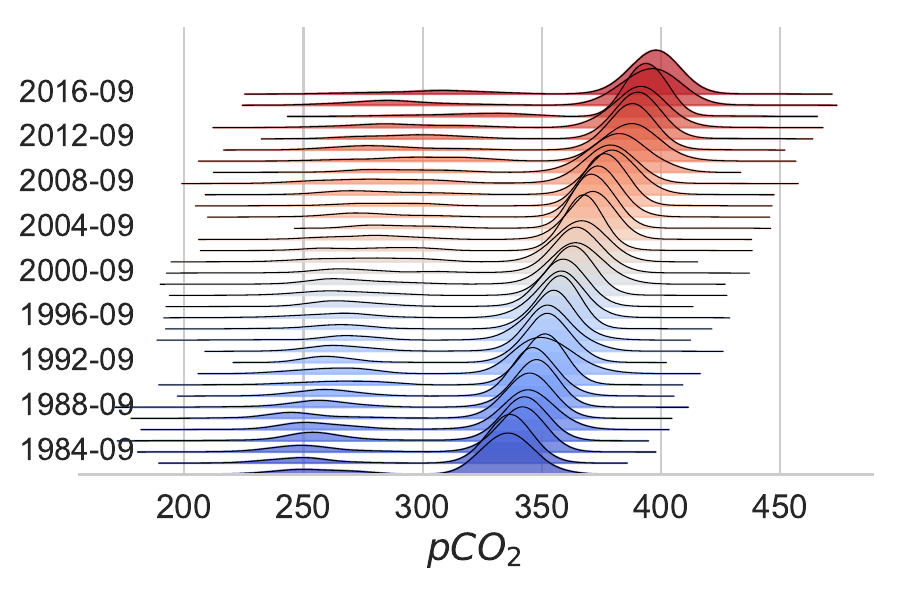}} \\
        \subcaptionbox{October\label{fig:oct_density}}{\includegraphics[width=0.31\textwidth]{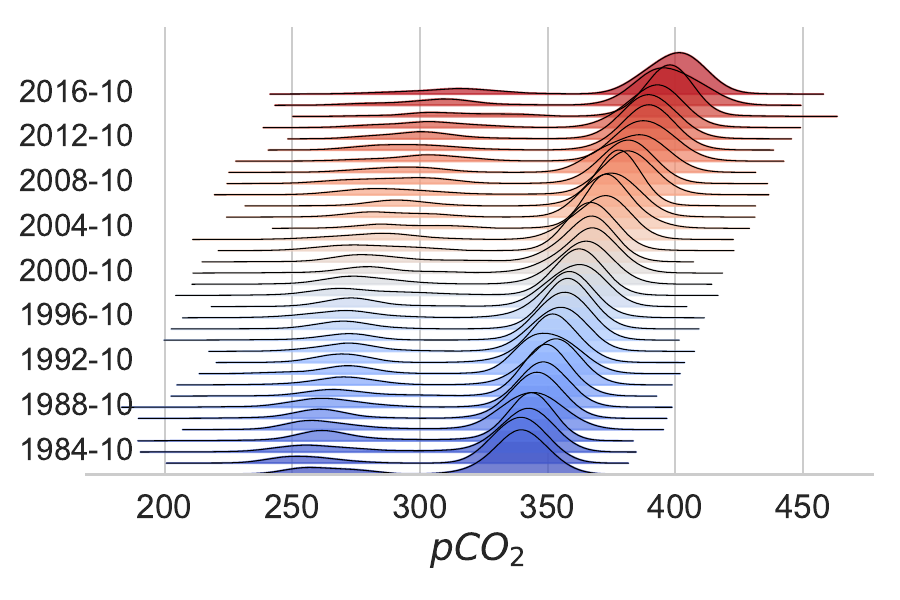}} &
        \subcaptionbox{November\label{fig:nov_density}}{\includegraphics[width=0.31\textwidth]{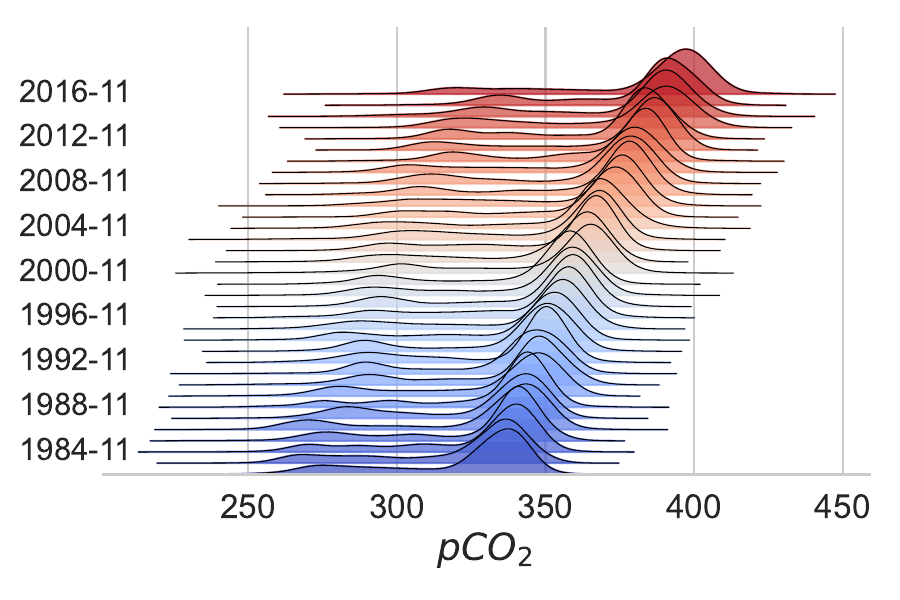}} &
        \subcaptionbox{December\label{fig:dec_density}}{\includegraphics[width=0.31\textwidth]{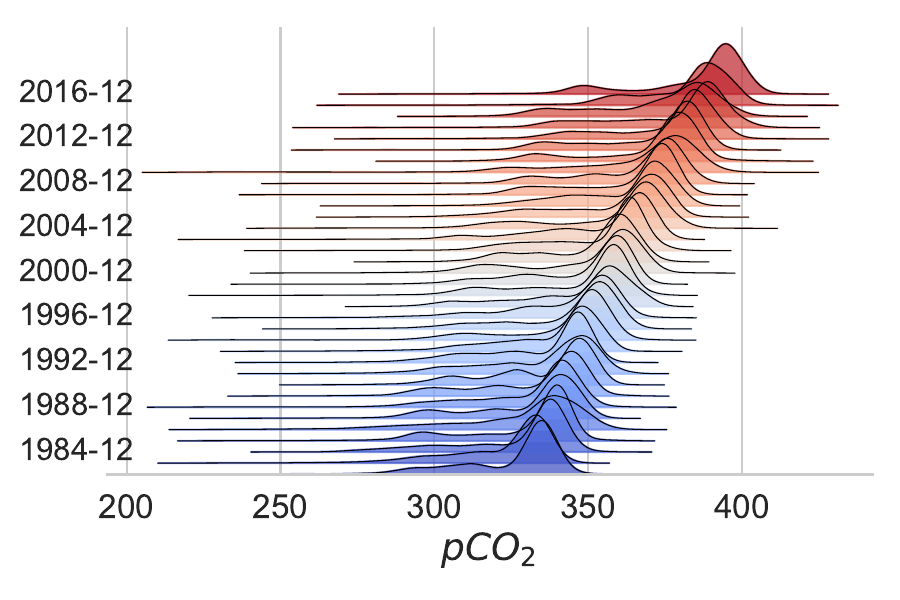}} \\
    \end{tabular}
    
    \caption{Monthly knowledge distributions of pCO\textsubscript{2} of the 6th province from CESM002. Each panel corresponds to a calendar month, and each horizontal ridge shows the marginal distribution of pCO\textsubscript{2} in a given year for that month over 1982--2016, with the vertical axis indicating year. The monthly distributions exhibit strong seasonal structure: while several winter and spring months appear approximately unimodal, late-summer and autumn months (July--November) display a clearer bimodal pattern, with distinct low and high pCO\textsubscript{2} modes. For any fixed month, the ridges drift gradually to the right over time, reflecting the long-term increase in pCO\textsubscript{2}, yet the overall shape of the distribution remains relatively stable across years. These patterns motivate the use of CESM002 as a knowledge source for KSD calibration: the model provides month-specific marginal distributions that encode robust seasonal and multimodal structure and long-term trends, which can be used to correct systematic biases in data-driven predictions trained on sparse and biased observations.}

    \label{fig:pco2_monthly_distribution}
\end{figure}

To assess the utility of KSD calibration using prior knowledge of the marginal outcome distribution, we view climate system models as providing approximate but structurally realistic priors for province--level pCO\textsubscript{2}. In our main experiments, we extract \emph{knowledge distributions} from a different ensemble member, CESM002, which is generated by the same CESM RCP8.5 configuration as CESM001 but with different initial values and boundary conditions. As a result, CESM001 and CESM002 share nearly identical seasonal and distributional structure, while remaining distinct realizations of the underlying climate system. In order to get the knowledge distribution, for each calendar month, we collect pCO\textsubscript{2} values of the 6th province from CESM002 over 2010-2016 and fit a Gaussian Mixture Model (GMM) with 5 components. The resulting month-specific mixtures provide flexible summaries of the marginal distributions, capturing both multi-modality and tail behavior.


Figure~\ref{fig:pco2_monthly_distribution} visualizes these month-specific knowledge distributions. Each panel corresponds to a calendar month, and each horizontal ridge within a panel represents the marginal pCO\textsubscript{2} distribution in a given year of the CESM002 time series for that month. Two robust features are apparent. First, different months exhibit distinct distributional shapes: while many winter and spring months look approximately unimodal, late-summer and autumn months (July--December) display a clearer bimodal structure, with both low and high pCO\textsubscript{2} modes present. Second, for any fixed month, the ridges shift gradually to the right over 1982--2016, reflecting the long-term increase in pCO\textsubscript{2}, but the overall shape of the distribution (for example, the number and relative prominence of modes) remains relatively stable from year to year. In this sense, the CESM-based knowledge distributions encode both a stable, month-specific seasonal structure and a slowly varying trend that are difficult to infer reliably from sparse and biased SOCAT observations alone. These are precisely the kinds of distributional features that our KSD-based calibration can exploit, because the KSD objective compares the full predictive distribution to the knowledge distribution and is sensitive to discrepancies in shape (such as missing modes or incorrect tail behavior), not just in the mean. By calibrating month by month, we explicitly respect the strong seasonal differences in the pCO\textsubscript{2} marginals instead of imposing a single global prior. Moreover, the knowledge distributions only need to be structurally correct rather than exact: other climate models, such as GFDL, produce monthly pCO\textsubscript{2} distributions with similar seasonal and multimodal structure, and could in principle be used in the same way to calibrate CESM001 or other targets.

We use these knowledge distributions from CESM002 to calibrate the raw predictions for CESM001 from 2010-2016, based on FFN trained on sparse data with sampling bias from 1982-2009. The calibration is repeated for each month, using the corresponding marginal distribution of the month. See the Appendix~\ref{pco2appendix} for the full implementation details of this experiment.

\subsubsection{Results}

\begin{table}[h]
\vspace{-.3cm}
    \caption{Evaluation of raw, normalized, and KSD-calibrated predictions of pCO\textsubscript{2} in the 6th province. The table reports the mean squared error (MSE), which measures point-wise prediction accuracy, and the 2-WD\textsuperscript{2}, which quantifies the discrepancy between the predicted and true marginal distributions of pCO\textsubscript{2}. The \textit{raw} column corresponds to the original FFN predictions trained on sparse and biased SOCAT observations; the \textit{normalized} column applies a simple rescaling step from {\bf Sec.}~\ref{sec:calibration} to put the outputs on a common marginal scale; and the \textit{calibrated} column applies our KSD-based calibration on top of this baseline. The KSD-calibrated model achieves the lowest MSE and 2-WD\textsuperscript{2}, with particularly large relative gains in 2-WD\textsuperscript{2}, showing that our KSD procedure is key to improving both point-wise accuracy and the alignment of the predictive distribution with the true pCO\textsubscript{2} distribution.}
\label{tab:pco2_calibration}
    \centering
    \begin{tabular}{cccccc}\hline
        \textbf{metric} & \textbf{raw} & \textbf{normalized} & \textbf{calibrated} \\\hline
        \textbf{MSE} & 703.7 & 421.8 & 400.8 \\
        \textbf{2-WD\textsuperscript{2}} & 411.3 & 31.0 & 23.6 \\\hline
    \end{tabular}
\vspace{-.3cm}
\end{table}

The overall performance is summarized in Table~\ref{tab:pco2_calibration}, which reports both the mean squared error (MSE) and the 2-WD\textsuperscript{2} between the predicted and true pCO\textsubscript{2} distributions. We compare three versions of the FFN outputs: the \textit{raw} predictions, a \textit{normalized} version obtained by a simple rescaling step in {\bf Sec.}~\ref{sec:calibration}, and the \textit{calibrated} predictions after applying our KSD-based distributional calibration using the CESM002 knowledge distributions. Relative to the raw model, the KSD-calibrated predictions reduce the MSE from 703.7 to 400.8 (about a 43\% reduction) and the 2-WD\textsuperscript{2} from 411.3 to 23.6 (about a 94\% reduction), indicating substantial gains in both point-wise accuracy and distributional fidelity. Even when starting from the normalized baseline, which mainly serves to place the FFN outputs on a common scale, the KSD step further decreases 2-WD\textsuperscript{2} from 31.0 to 23.6. This additional reduction highlights that KSD calibration is crucial for correcting higher-order distributional features—beyond simple rescaling—and for bringing the predictive marginal much closer to the true pCO\textsubscript{2} distribution.

\begin{figure}[htbp!]
    \centering
    \begin{tabular}{ccc}
        \subcaptionbox{January\label{fig:jan_cali}}{\includegraphics[width=0.31\textwidth]{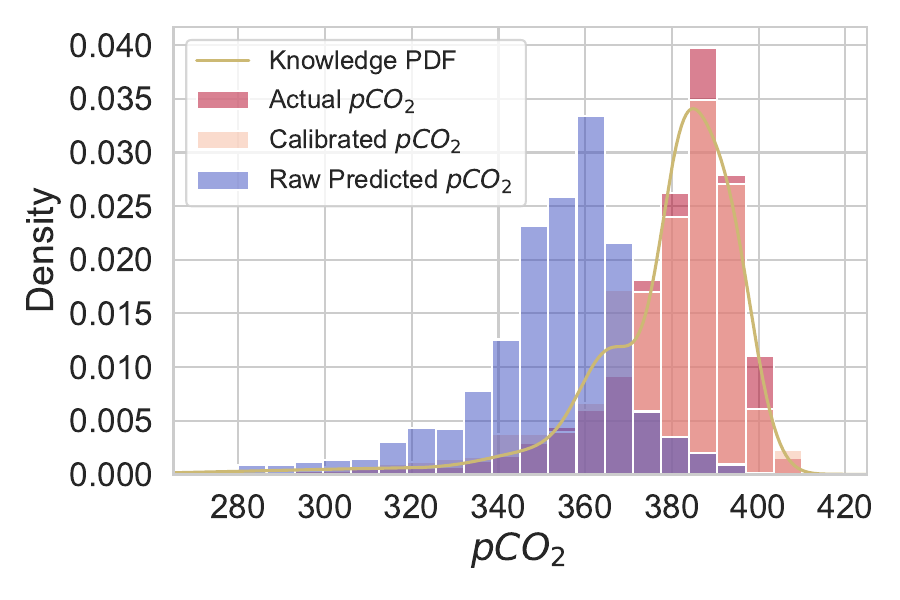}} &
        \subcaptionbox{February\label{fig:feb_cali}}{\includegraphics[width=0.31\textwidth]{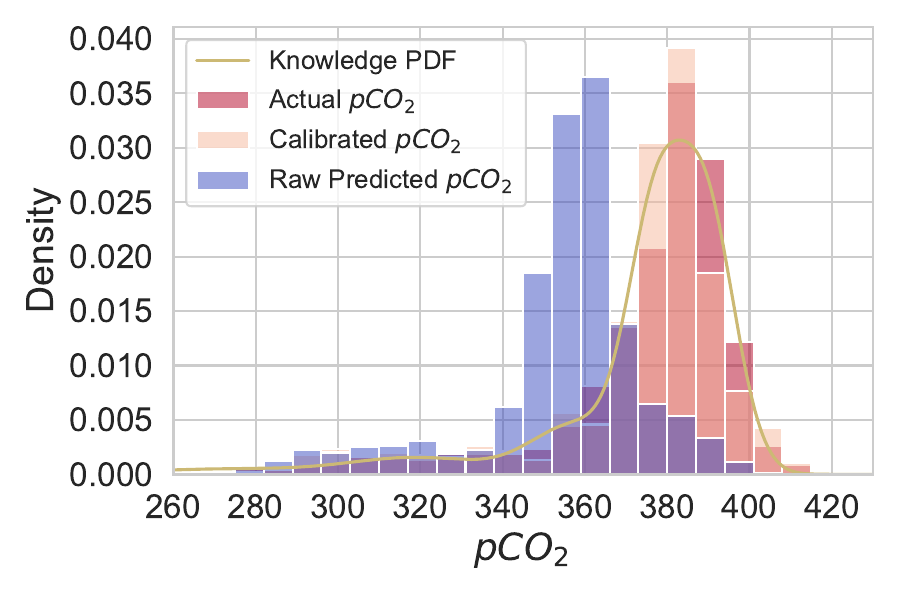}} &
        \subcaptionbox{March\label{fig:mar_cali}}{\includegraphics[width=0.31\textwidth]{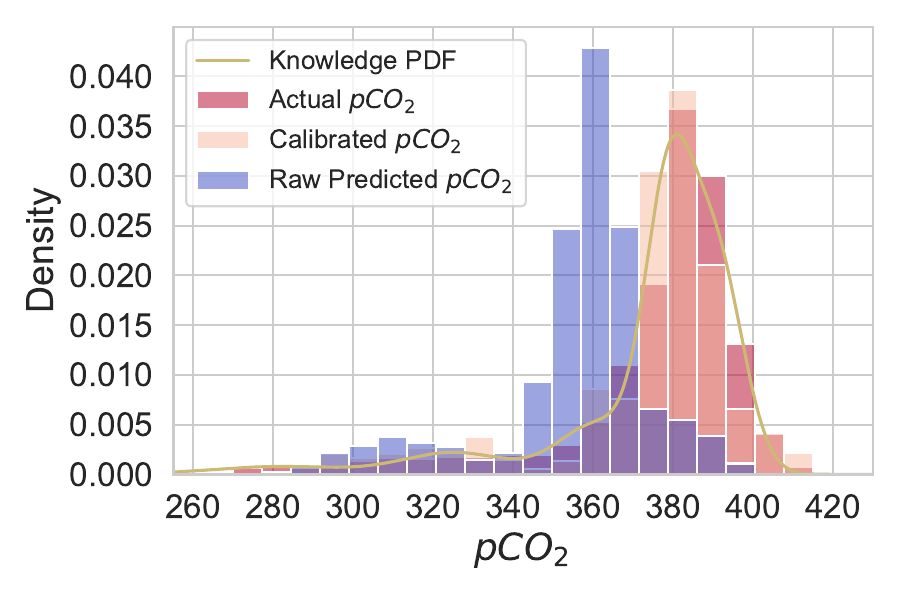}} \\
        \subcaptionbox{April\label{fig:apr_cali}}{\includegraphics[width=0.31\textwidth]{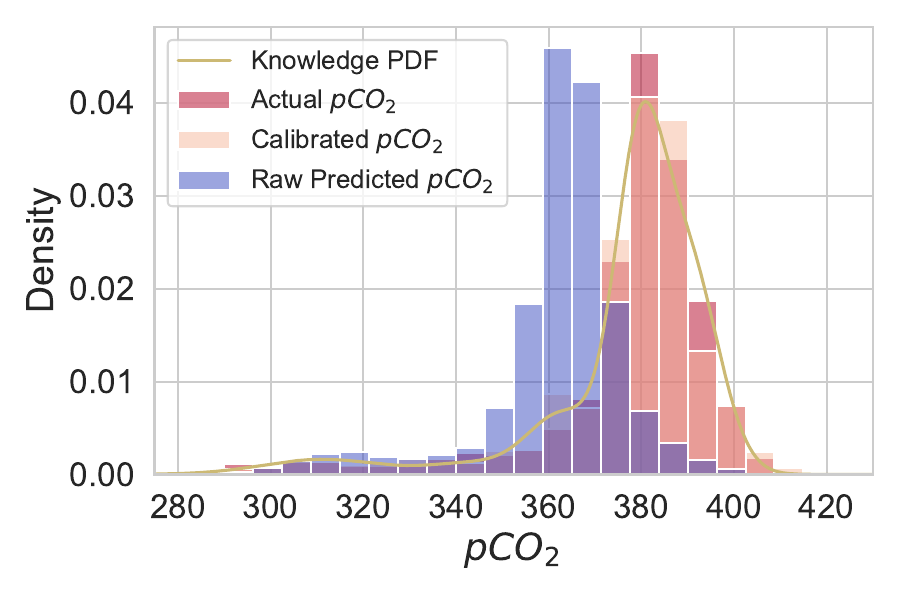}} &
        \subcaptionbox{May\label{fig:may_cali}}{\includegraphics[width=0.31\textwidth]{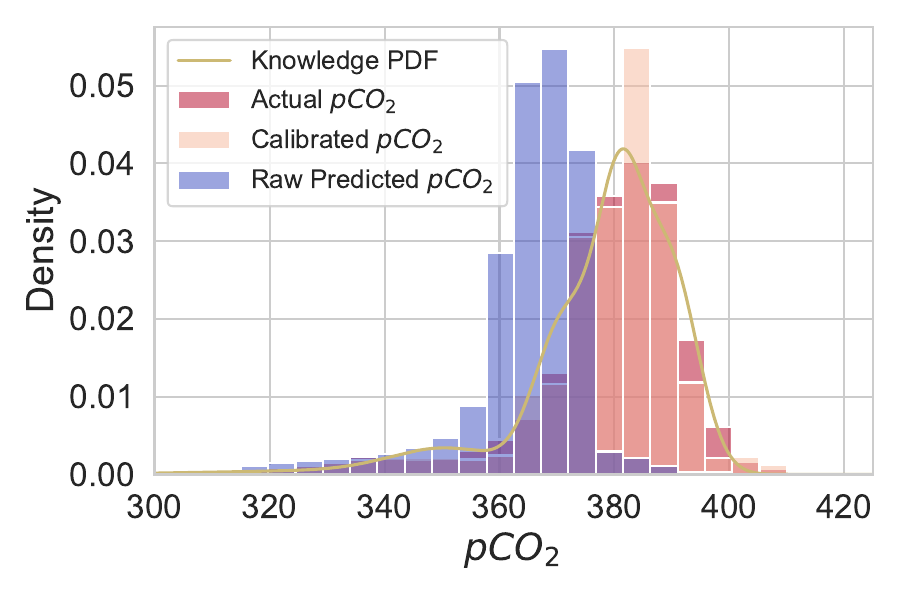}} &
        \subcaptionbox{June\label{fig:jun_cali}}{\includegraphics[width=0.31\textwidth]{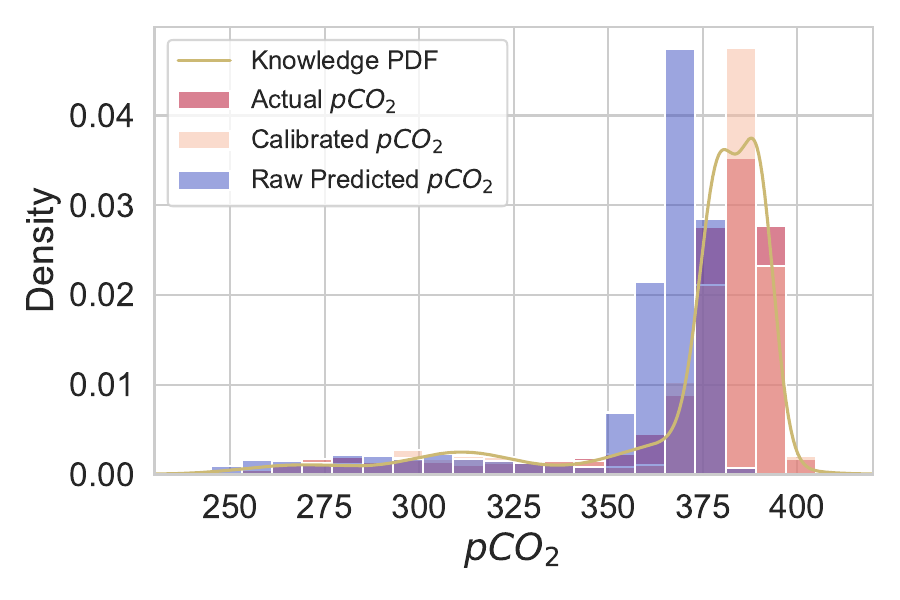}} \\
        \subcaptionbox{July\label{fig:jul_cali}}{\includegraphics[width=0.31\textwidth]{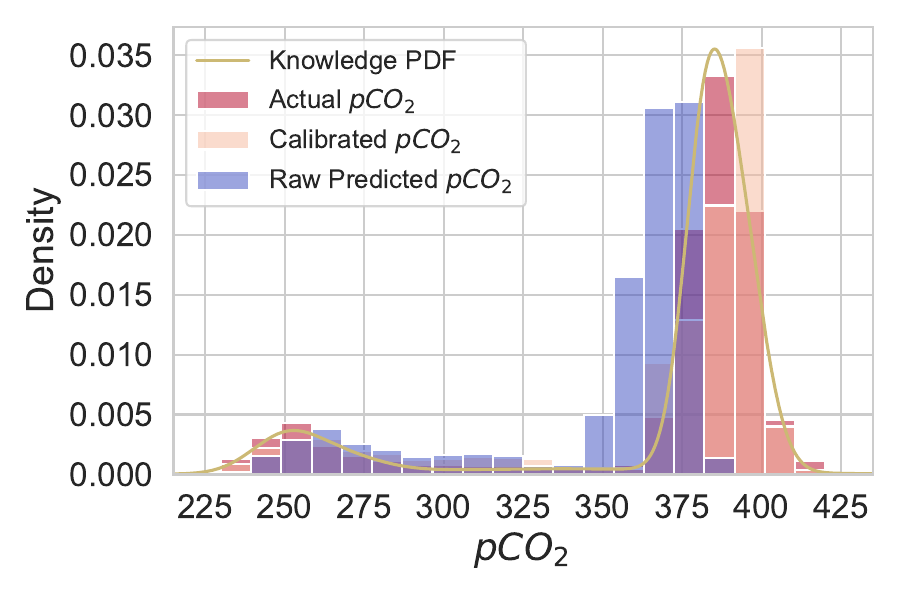}} &
        \subcaptionbox{August\label{fig:aug_cali}}{\includegraphics[width=0.31\textwidth]{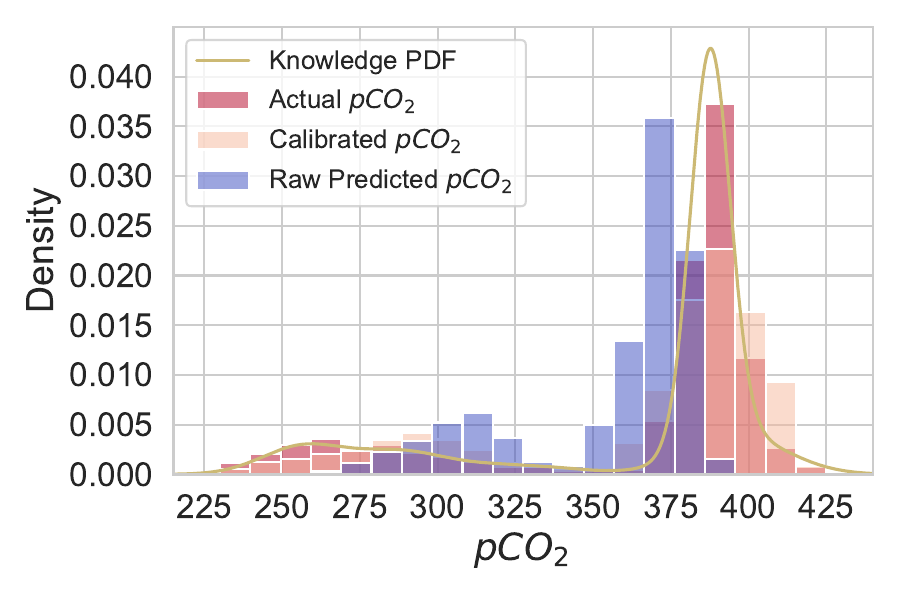}} &
        \subcaptionbox{September\label{fig:sep_cali}}{\includegraphics[width=0.31\textwidth]{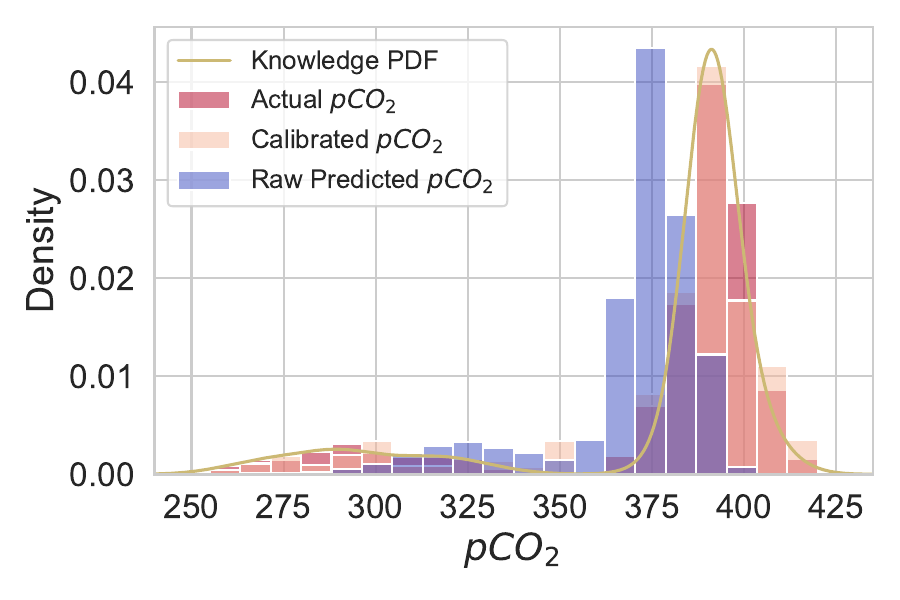}} \\
        \subcaptionbox{October\label{fig:oct_cali}}{\includegraphics[width=0.31\textwidth]{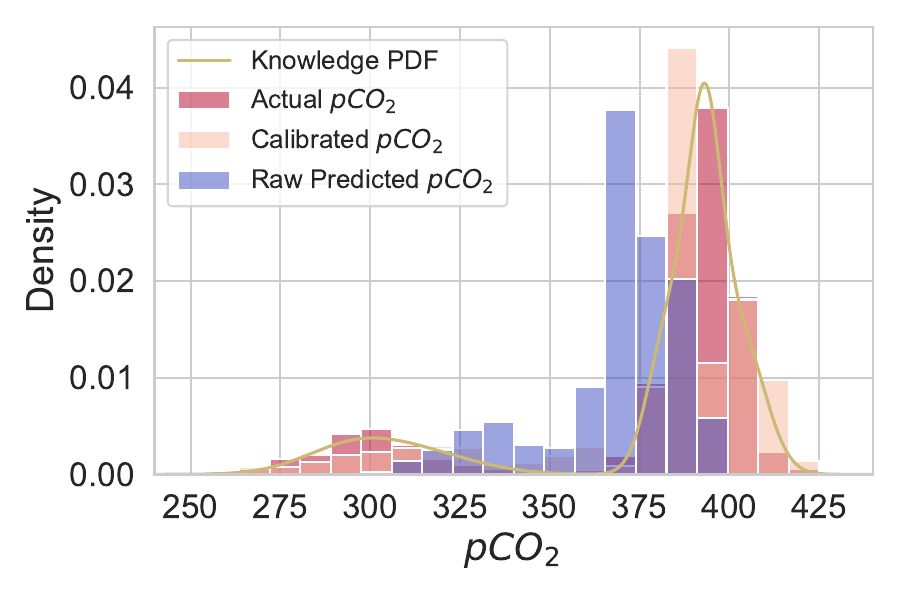}} &
        \subcaptionbox{November\label{fig:nov_cali}}{\includegraphics[width=0.31\textwidth]{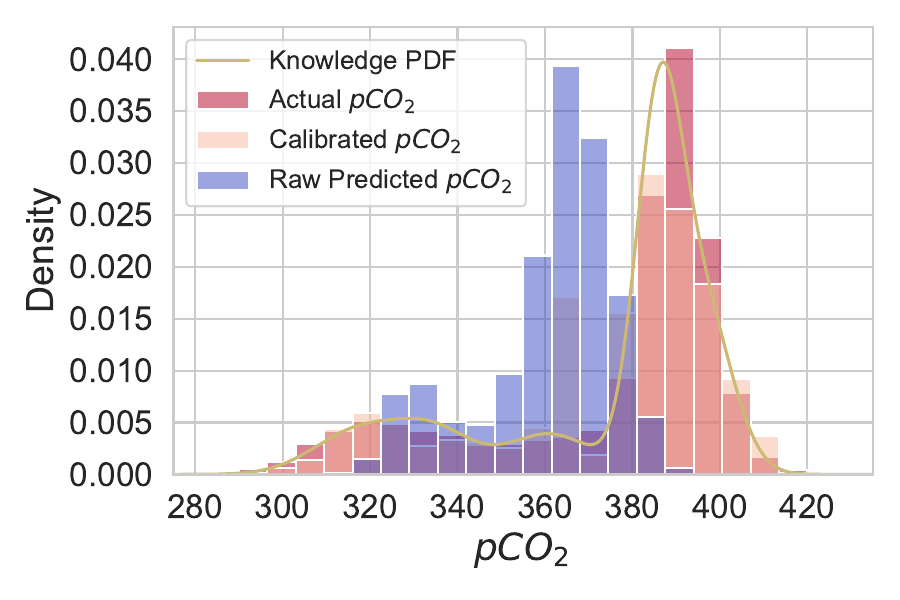}} &
        \subcaptionbox{December\label{fig:dec_cali}}{\includegraphics[width=0.31\textwidth]{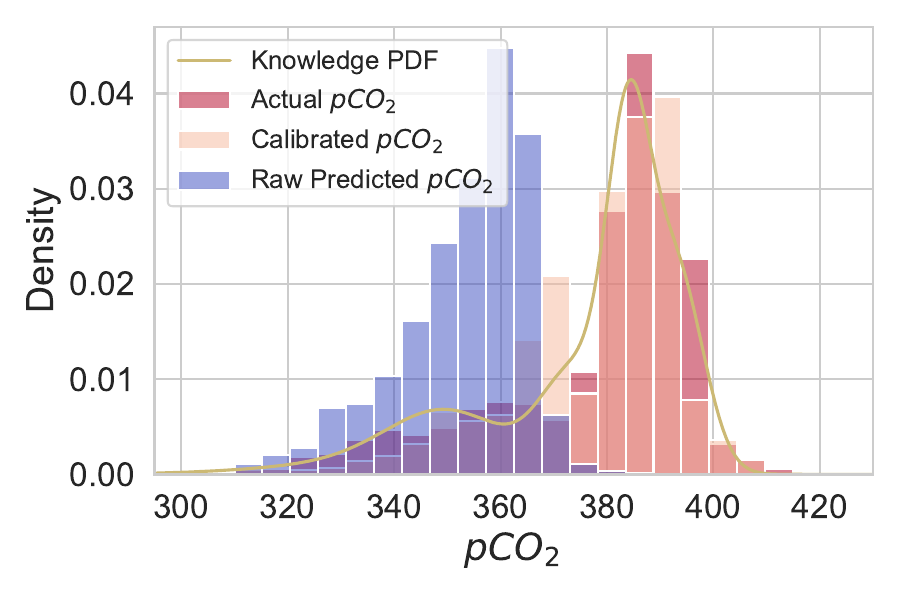}} \\
    \end{tabular}
    \caption{Monthly calibration of pCO\textsubscript{2} predictions for the 6th province in 2010-2016. Each panel corresponds to one calendar month and shows four distributions of pCO\textsubscript{2} values: the empirical distribution of the true CESM001 test data (red histogram), the distribution of raw FFN predictions (blue histogram), the distribution of KSD-calibrated predictions (orange histogram), and the month-specific knowledge distribution from CESM002 (yellow curve) obtained from a Gaussian mixture model. Although the marginal distributions differ from month to month, the CESM-based knowledge distributions and the true CESM001 marginals share similar overall structure, providing a physically informed target for calibration. In many months, the raw predictions are biased or miss important features of the true distribution—for example, they can be overly narrow, shifted relative to the true values, or fail to represent multi-modal behavior when present. After KSD calibration, the orange histograms generally move closer to the red ones and better track both the main modes and the range of observed pCO\textsubscript{2}, while remaining consistent with the yellow knowledge curves. Together with the quantitative results in Tables~\ref{tab:pco2_calibration}, \ref{tab:pco2_monthly_calibration_MSE} and \ref{tab:pco2_monthly_calibration_2WD2}, these panels show that KSD calibration improves the match between the predictive marginals and the true pCO\textsubscript{2} distributions across all months.
    }
    \label{fig:pco2_monthly_calibration}
\end{figure}

Figure~\ref{fig:pco2_monthly_calibration} provides a month-by-month view of how KSD calibration reshapes the predictive distributions of pCO\textsubscript{2} in the 6th province. For each calendar month, we show the empirical distribution of the true CESM001 test values (red histogram), the raw FFN predictions (blue histogram), the KSD-calibrated predictions (orange histogram), and the corresponding knowledge distribution from CESM002 (yellow curve) obtained from the Gaussian mixture model described above. Because CESM001 and CESM002 are generated by the same climate model configuration, their monthly marginals share very similar structure, so the knowledge distributions provide a realistic, model-based description of the expected shape of the pCO\textsubscript{2} distribution in this province.

Across months, the raw FFN predictions often fail to reproduce key distributional features of the true pCO\textsubscript{2} marginals: they may be overly concentrated in a narrow range, shifted relative to the red histogram, or unable to capture multiple modes when the actual and knowledge distributions display more complex shapes. After applying KSD calibration, the orange histograms move closer to the red ones while remaining consistent with the yellow knowledge curves, with better alignment of the main peaks, support, and overall spread. In other words, KSD uses the CESM-based knowledge distributions to correct higher-order aspects of the predictive marginal—beyond simple shifts or rescaling—so that, for each month, the calibrated predictions more faithfully reflect the underlying pCO\textsubscript{2} distribution. This figure thus complements the aggregate metrics in Table~\ref{tab:pco2_calibration},~\ref{tab:pco2_monthly_calibration_MSE} and~\ref{tab:pco2_monthly_calibration_2WD2} by illustrating, in distributional form, how calibration improves both the shape and location of the predictive marginals.

Table~\ref{tab:pco2_monthly_calibration_MSE},~\ref{tab:pco2_monthly_calibration_2WD2} summarize the monthly pCO\textsubscript{2} statistics for the 6th province. The \textit{Location Num.} column indicates the total number of location points, while the \textit{Obv.} column shows the number of distinct observations recorded each month. From the data, it is evident that the calibrated MSE values are consistently lower than the raw MSE values, indicating that the calibration process improves the accuracy of the pCO\textsubscript{2} predictions. The 2-WD\textsuperscript{2} values, which measure the distance between the predicted and true distributions, also decrease significantly after calibration.

\begin{table}[h]
    \caption{Monthly MSE of raw, normalized, and KSD-calibrated pCO\textsubscript{2} predictions in the 6th province. For each calendar month, the table reports the mean squared error (MSE) of the raw FFN predictions, the normalized outputs, and the KSD-calibrated predictions. The \textit{Location Num.} column gives the total number of distinct 1\textdegree$\times$1\textdegree{} grid cells that belong to the 6th province, and the \textit{Obv.} column records the number of distinct SOCAT observations falling in the province in that month. The \textit{Avg} row reports totals for location numbers and observation counts and monthly averages for the three MSE values. Across all twelve months, the calibrated MSE is always lower than or comparable to the raw MSE, with the largest relative reductions occurring in months where the baseline FFN exhibits substantial bias, illustrating that the full calibration pipeline improves point-wise predictive accuracy over a range of data availability regimes.
    }
    
    \label{tab:pco2_monthly_calibration_MSE}
    
    \centering
    \begin{tabular}{rrrrrr}
        \toprule
        Month & Location Num. & Obv. & Raw MSE & Normalized MSE & Calibrated MSE \\
        \midrule
        Jan & 15263 & 458 & 1202.47 & 494.63 & 474.17 \\
        Feb & 16273 & 733 & 784.43 & 453.56 & 387.85 \\
        Mar & 17837 & 767 & 661.13 & 340.86 & 302.94 \\
        \midrule
        Apr & 17491 & 642 & 442.15 & 253.52 & 239.57 \\
        May & 16534 & 708 & 429.49 & 268.77 & 277.54 \\
        Jun & 16130 & 496 & 498.69 & 267.08 & 278.04 \\
        \midrule
        Jul & 15478 & 376 & 694.17 & 453.21 & 452.95 \\
        Aug & 14618 & 218 & 875.93 & 699.88 & 706.89 \\
        Sep & 12671 & 127 & 569.71 & 451.00 & 388.97 \\
        \midrule
        Oct & 12980 & 89 & 588.67 & 412.73 & 347.88 \\
        Nov & 15039 & 241 & 739.91 & 664.72 & 629.00 \\
        Dec & 15443 & 352 & 996.85 & 378.64 & 387.77 \\
        \midrule
        Avg & 185757 & 5207 & 703.71 & 421.83 & 400.80 \\
        \bottomrule
    \end{tabular}
    
\end{table}

\begin{table}[h!]
    \caption{Monthly 2-WD\textsuperscript{2} of raw, normalized, and KSD-calibrated pCO\textsubscript{2} predictions in the 6th province. For each calendar month, the table reports the 2-WD\textsuperscript{2} distance between the predicted and true marginal distributions of pCO\textsubscript{2} for the raw FFN outputs, the normalized predictions, and the KSD-calibrated predictions. As in Table~\ref{tab:pco2_monthly_calibration_MSE}, the \textit{Location Num.} column gives the number of distinct 1\textdegree$\times$1\textdegree{} grid cells assigned to the 6th province and the \textit{Obv.} column records the number of distinct SOCAT observations in that month, while the \textit{Avg} row aggregates counts and monthly averages. In every month, 2-WD\textsuperscript{2} decreases as we move from raw to normalized to KSD-calibrated predictions, with particularly large relative reductions in data-sparse months such as September and October. These patterns show that KSD-based calibration systematically improves the alignment between the predictive marginals and the true pCO\textsubscript{2} distributions across all months.}
    
    \label{tab:pco2_monthly_calibration_2WD2}
    
    \centering
    \begin{tabular}{rrrrrr}
        \toprule
        Month & Location Num. & Obv. & Raw 2-WD\textsuperscript{2} & Normalized 2-WD\textsuperscript{2} & Calibrated 2-WD\textsuperscript{2} \\
        \midrule
        Jan & 15263 & 458 & 766.45 & 26.68 & 25.43 \\
        Feb & 16273 & 733 & 441.75 & 29.06 & 12.43 \\
        Mar & 17837 & 767 & 415.08 & 29.81 & 19.14 \\
        \midrule
        Apr & 17491 & 642 & 256.28 & 14.56 & 6.72 \\
        May & 16534 & 708 & 197.79 & 2.98 & 1.49 \\
        Jun & 16130 & 496 & 244.16 & 10.54 & 9.21 \\
        \midrule
        Jul & 15478 & 376 & 350.69 & 32.14 & 30.48 \\
        Aug & 14618 & 218 & 466.02 & 56.56 & 55.70 \\
        Sep & 12671 & 127 & 329.52 & 37.45 & 25.30 \\
        \midrule
        Oct & 12980 & 89 & 410.73 & 82.09 & 51.82 \\
        Nov & 15039 & 241 & 377.02 & 27.40 & 20.62 \\
        Dec & 15443 & 352 & 694.56 & 16.31 & 15.96 \\
        \midrule
        Avg & 185757 & 5207 & 409.81 & 29.15 & 21.85 \\
        \bottomrule
    \end{tabular}

\end{table}

The location number is relatively stable across months, but the number of distinct SOCAT observations varies substantially, and this heterogeneity in observational density is reflected in both the MSE and 2-WD\textsuperscript{2} values.
Months with relatively dense observational coverage such as March, April and May already have moderate raw MSE values and still exhibit sizable reductions after the normalization and calibration steps. By contrast, late-summer and early-fall months such as August, September, and October have far fewer distinct observations despite a comparable number of grid points, and start from relatively large raw errors. Even in these data-sparse settings the full pipeline yields clear gains: for example, in October the MSE decreases from 588.67 to 347.88 and the 2-WD\textsuperscript{2} drops from 410.73 to 51.82 after normalization and KSD calibration.

Across all twelve months, the MSE is reduced by the end of the pipeline, with the average monthly MSE falling from 703.71 for the raw FFN predictions to 400.80 after calibration. The normalization step lowers MSE in every month, often by several hundred units, and the KSD calibration step further decreases MSE for most months. In a few cases, KSD calibration slightly increases MSE relative to the normalized predictions while still improving it relative to the raw FFN output, indicating a small trade-off between point-wise accuracy and distributional alignment. The largest relative reductions in MSE occur in months where the raw FFN predictions are most biased (e.g., January and December, with more than 60\% reductions), whereas months in which the baseline FFN already performs relatively well (such as May and November) show more modest MSE gains.

For the distributional metric, improvements are even more uniform: 2-WD\textsuperscript{2} decreases monotonically from the raw to normalized to calibrated predictions in every month, with the average value dropping from 409.81 to 21.85. Normalization alone accounts for the bulk of the reduction, bringing 2-WD\textsuperscript{2} down to the tens across all months, and KSD calibration consistently provides an additional refinement. The relative gains are particularly striking in months with very sparse SOCAT coverage such as September and October, where raw 2-WD\textsuperscript{2} values above 300 are reduced by nearly an order of magnitude after the full procedure. This pattern suggests that prior knowledge about the marginal distribution of pCO\textsubscript{2} from CESM002 is especially valuable when local observations are scarce, allowing KSD-based calibration to substantially correct the distribution of FFN predictions even in poorly observed months.

\subsubsection{Robustness Analysis}

\begin{table}[htbp!]
    \caption{Leave-one-out stability of monthly MSE for pCO\textsubscript{2} predictions in the 6th province. For each calendar month, the table reports the mean and standard deviation of the MSE over leave-one-out (LOO) evaluations on the 2010--2016 period, where one evaluation year is omitted at a time and the metric is computed on the remaining six years. The three columns correspond to the raw FFN predictions, the normalized outputs, and the KSD-calibrated predictions. In every month, the calibrated MSE is lower than the raw MSE, and in most months it is also lower than the normalized MSE, indicating that the benefits of calibration extend to month-specific performance rather than being driven by a few favorable years. The reported standard deviations are moderate across months, suggesting that the relative performance ordering of the three methods is stable under different choices of held-out evaluation year.}
    \label{tab:pco2_stability}
    \centering
    \begin{tabular}{cccc}\hline
        \textbf{metric} & \textbf{raw} & \textbf{normalized} & \textbf{calibrated} \\\hline
        \textbf{MSE} & 703.8 (45.4) & 421.8 (14.2) & 400.8 (14.2) \\
        \textbf{2-WD\textsuperscript{2}} & 411.9 (41.1) & 31.8 (3.2) & 24.4 (2.9) \\\hline
    \end{tabular}
\end{table}

We further assess the robustness of the performance metrics using a leave-one-out (LOO) evaluation approach. Specifically, for the evaluation period 2010--2016, we compute the MSE and 2-WD\textsuperscript{2} on all subsets of six years obtained by excluding one year at a time. Table~\ref{tab:pco2_stability} reports the mean and standard deviation of these LOO metrics for each method. The LOO means are nearly identical to the full-sample values in Table~\ref{tab:pco2_calibration}, and the standard deviations are small, indicating that the improvements are not driven by any single evaluation year and are stable with respect to interannual variability in the test period.

\begin{table}[htbp!]
    \caption{Leave-one-out robustness of aggregate pCO\textsubscript{2} performance metrics. Entries show the mean and standard deviation of each metric over leave-one-out (LOO) evaluations on the 2010--2016 period, where in each run one year is excluded and the metrics are computed on the remaining six years. The \textit{raw}, \textit{normalized}, and \textit{calibrated} columns correspond to the original FFN predictions, the normalized outputs, and the KSD-calibrated predictions, respectively. The LOO means are very close to the full-sample values reported in Table~\ref{tab:pco2_monthly_calibration_MSE}, and the standard deviations are small, indicating that the performance differences between methods are stable across different choices of evaluation years.}
    
    \label{tab:pco2_monthly_stability_MSE}
    
    \centering
    \begin{tabular}{rrrr}
        \toprule
        Month & Raw MSE & Normalized MSE & Calibrated MSE \\
        \midrule
        Jan & 1202.52 (89.25) & 494.62 (29.19) & 474.17 (29.17) \\
        Feb & 784.42 (60.35) & 453.58 (22.03) & 387.88 (22.95) \\
        Mar & 660.95 (51.12) & 340.78 (16.21) & 302.87 (17.35) \\
        \midrule
        Apr & 442.10 (32.26) & 253.46 (8.89) & 239.51 (8.55) \\
        May & 429.48 (27.64) & 268.74 (10.38) & 277.51 (11.52) \\
        Jun & 498.85 (41.21) & 267.07 (11.60) & 278.02 (12.27) \\
        \midrule
        Jul & 694.29 (45.29) & 453.13 (30.06) & 452.88 (28.37) \\
        Aug & 875.87 (45.38) & 699.66 (22.78) & 706.67 (23.31) \\
        Sep & 569.89 (24.94) & 451.06 (11.72) & 389.01 (9.16) \\
        \midrule
        Oct & 588.78 (22.07) & 412.76 (12.79) & 347.91 (10.36) \\
        Nov & 740.02 (39.90) & 664.75 (26.97) & 629.04 (27.43) \\
        Dec & 997.07 (75.93) & 378.70 (20.03) & 387.82 (19.90) \\
        \midrule
        Avg & 703.77 (45.36) & 421.81 (14.21) & 400.78 (14.19) \\
        \bottomrule
    \end{tabular}
\end{table}

\begin{table}[htbp!]
    \caption{Leave-one-out stability of monthly 2-WD\textsuperscript{2} for pCO\textsubscript{2} predictions in the 6th province. For each calendar month, the table reports the mean and standard deviation of the 2-WD\textsuperscript{2} distance between the predicted and true marginal pCO\textsubscript{2} distributions over leave-one-out (LOO) evaluations on 2010--2016. As in Table~\ref{tab:pco2_monthly_stability_MSE}, the three columns correspond to the raw FFN, normalized, and KSD-calibrated models. For every month, the mean 2-WD\textsuperscript{2} decreases monotonically from raw to normalized to calibrated, demonstrating that KSD consistently improves distributional alignment on top of the normalization step. The associated standard deviations remain small relative to the mean values, indicating that these distributional gains persist across different choices of held-out year and are not artifacts of a single evaluation split.}
    
    \label{tab:pco2_monthly_stability_2WD2}
    
    \centering
    \begin{tabular}{rrrr}
        \toprule
        Month & Raw 2-WD\textsuperscript{2} & Normalized 2-WD\textsuperscript{2} & Calibrated 2-WD\textsuperscript{2} \\
        \midrule
        Jan & 768.77 (70.79) & 28.95 (7.44) & 27.64 (7.33) \\
        Feb & 444.03 (54.70) & 31.92 (9.19) & 15.20 (6.49) \\
        Mar & 416.93 (47.65) & 32.30 (6.49) & 21.56 (6.15) \\
        \midrule
        Apr & 257.72 (31.04) & 16.42 (2.05) & 8.49 (2.71) \\
        May & 199.41 (31.61) & 4.80 (1.75) & 3.43 (1.62) \\
        Jun & 247.47 (39.76) & 13.51 (2.14) & 12.36 (2.39) \\
        \midrule
        Jul & 354.31 (50.61) & 36.39 (5.76) & 34.80 (3.87) \\
        Aug & 467.75 (40.96) & 59.53 (6.74) & 58.68 (7.02) \\
        Sep & 331.31 (23.63) & 40.41 (9.03) & 27.95 (5.79) \\
        \midrule
        Oct & 411.83 (21.75) & 84.40 (5.89) & 53.59 (4.97) \\
        Nov & 378.53 (27.78) & 29.84 (4.35) & 22.89 (4.15) \\
        Dec & 696.74 (61.26) & 18.67 (3.18) & 18.26 (3.15) \\
        \midrule
        Avg & 411.90 (41.10) & 31.77 (3.16) & 24.40 (2.92) \\
        \bottomrule
    \end{tabular}
\end{table}

Tables~\ref{tab:pco2_monthly_stability_MSE} and \ref{tab:pco2_monthly_stability_2WD2} provide more detailed LOO results at the monthly level. Across all months, the calibrated models consistently reduce both MSE and 2-WD\textsuperscript{2} relative to the raw FFN baseline. For 2-WD\textsuperscript{2}, the values decrease monotonically from raw to normalized to calibrated in every month, whereas for MSE the lowest value is sometimes attained by the normalized model and sometimes by the calibrated model. In most months, however, the calibrated MSE is smaller than both the raw and normalized MSE. Overall, the LOO means and standard deviations in these tables closely track those in Tables~\ref{tab:pco2_monthly_calibration_MSE} and \ref{tab:pco2_monthly_calibration_2WD2}, confirming that our main conclusions about the benefits of KSD calibration are robust to the choice of evaluation years.

\subsection{Application to Online Hybrid Emulator of QG Turbulence}

\subsubsection{Background and Motivation}

To assess the effectiveness of our calibration framework in realistic, high-dimensional settings, we apply it to a quasi-geostrophic (QG) ocean turbulence model \cite{pyqg} in an online deployment context. The QG model is a canonical representation of large-scale, stratified, rotating geophysical flows, capturing essential physical mechanisms such as nonlinear advection, planetary vorticity gradients (the $\beta$-effect), vortex stretching, and baroclinic instability. These processes govern mesoscale turbulence in the ocean and atmosphere and make QG dynamics a widely used testbed for studying the interaction between machine learning and chaotic physical systems.

A defining feature of QG turbulence is the inverse energy cascade, where energy injected at small scales (e.g., via baroclinic forcing) is transferred upscale, generating large, coherent structures such as mesoscale eddies and zonal jets. Accurate modeling of these features requires high spatial resolution. Coarse-resolution QG simulations often fail to resolve eddies, leading to unrealistic dynamics. To address this, \citet{ross2023benchmarking} proposed a fully connected neural network (FCNN) that learns a surrogate subgrid forcing term from low-resolution fields to recover unresolved effects. As discussed previously, this method fails to recover key derived spectral features of the modeled systems, which are fundamentally governed by turbulence physics based on first principles.

\subsubsection{Experiment Setup}

We evaluate KSD calibration in both offline and online settings. The goal is to assess how the proposed method improves the long-term stability and physical realism of emulators under chaotic flow dynamics.

Following \cite{ross2023benchmarking}, we carry out experiments under the \textit{eddy regime}, which is characterized by isotropic mesoscale eddies, resulting from baroclinic instability and vortex interactions. We compare the following four methods against the reference high-resolution simulation (High-res): 1) \textbf{Low-res}: a coarse-resolution simulation without correction; 2) \textbf{Low-res-FCNN}: the low-res model augmented with the learned FCNN subgrid forcing from \cite{ross2023benchmarking};
3) \textbf{Low-res-KSD}: the low-res model with KSD calibration; 4) \textbf{Low-res-FCNN-KSD}: the low-res-FCNN model with KSD calibration. This comprehensive comparison enables us to quantify both the standalone and complementary benefits of data-driven subgrid parameterization and distribution-informed post hoc calibration.

We simulate the high-resolution QG model on a $256 \times 256$ grid using a 1-hour time step. A spin-up period of 31,000 hours was used to ensure the system reaches statistical stationarity. These simulations serve as the ground truth. To construct the low-resolution model, following \cite{ross2023benchmarking}, we apply a spectral operator to project the high-resolution fields onto a $64 \times 64$ grid, preserving only the large-scale components. The low-res model is then run forward for 86,000 hours (approximately 10 years). KSD calibration is applied every 1,000 hours to adjust the prognostic variable $q$ at each vertical level. Snapshots are recorded every 1,000 hours for evaluation. The knowledge distribution used for calibration is constructed from an independent set of 87 low-resolution snapshots, by applying the same spectral operator to the high-res simulation over the same 86,000-hour period. For each vertical level, we aggregate spatial values of $q$ across all snapshots and fit a Gaussian distribution to serve as the target steady-state distribution for KSD calibration.

\subsubsection{Offline Performance}

To test the effectiveness of the approach in a simpler setting, we conduct an offline experiment. Each low-res snapshot at time $t$ is used as the initial condition for a 1,000-hour forward integration of the low-res model. The output at time $t+1000$ is then compared to the corresponding low-res ground truth. Calibration via KSD using the fitted steady-state distribution at $t+1000$ yields modest improvements in pointwise statistics: for example, level-1 $R^2$ increases from 0.533 to 0.541 and MSE decreases from $2.84 \times 10^{-11}$ to $2.79 \times 10^{-11}$; level-2 $R^2$ improves from 0.844 to 0.848 with a corresponding MSE drop from $1.81 \times 10^{-13}$ to $1.77 \times 10^{-13}$.

\subsubsection{Online Performance}

We directly apply the proposed KSD calibration in an online setting to align the state distribution of a low-resolution QG model with that of a high-resolution reference periodically. Rather than relying on a parametric closure during training, our method iteratively adjusts the model output to match the long-term statistical signature of the true system. KSD calibration leads to comparable or better accuracy. We further evaluate KSD calibration's skill to improve the online performance of the low-res model augmented with the learned subgrid forcing from \cite{ross2023benchmarking}. KSD calibration effectively improves online simulations across multiple diagnostics.

Five independent online experiments are conducted to ensure the statistical stability of our evaluation. See Appendix~\ref{pyqgappendix} for the full experiment implementation details.

\begin{figure}[htbp!]
\centering
\includegraphics[width=0.81\textwidth]{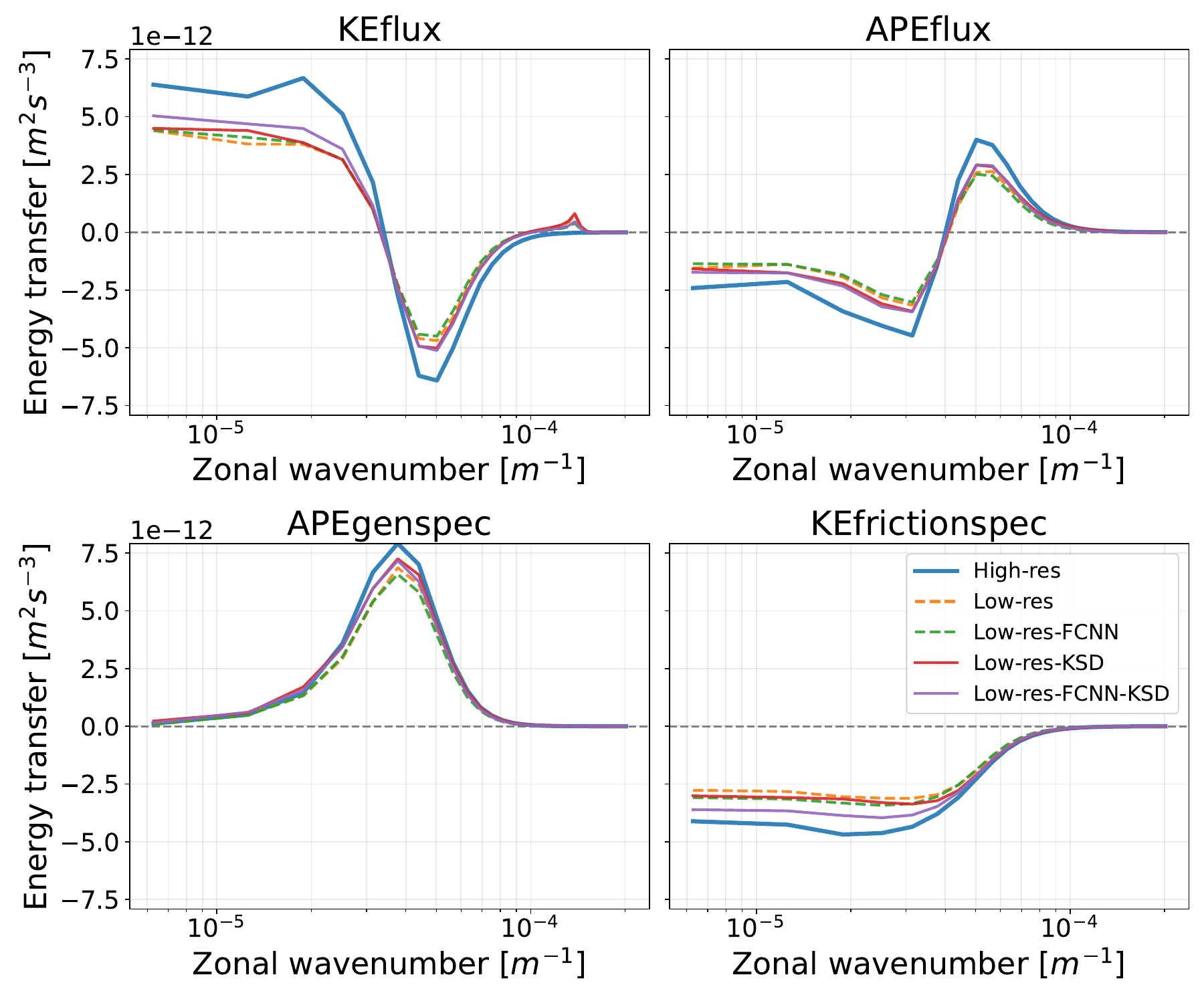} 
\caption{Online spectral energy diagnostics for five modeling setups in quasi-geostrophic turbulence. Each panel shows a derived statistical quantity as a function of zonal wavenumber: kinetic energy flux (KEflux), available potential energy flux (APEflux), APE generation spectrum (APEgenspec), and kinetic energy frictional dissipation spectrum (KEfrictionspec). The high-resolution simulation (blue solid line) provides the target spectra. The coarse low-resolution model (orange dashed line) and its FCNN-augmented variant (green dashed line) remain systematically biased relative to the high-resolution curves, with only modest differences between the orange and green spectra. In contrast, the KSD-calibrated low-resolution model (red solid line) substantially reduces these discrepancies: for KEflux, APEflux, and APEgenspec, the red curves lie much closer to the blue reference than either orange or green, and for KEfrictionspec they move in the correct direction relative to the uncalibrated models. The combined Low-res-FCNN-KSD model (purple solid line) achieves the closest overall match to the high-resolution spectra, nearly coinciding with the blue curves across most wavenumbers. Together with Figure~\ref{fig:pyqg}, this figure illustrates that KSD calibration is effective at restoring physically realistic long-term energy statistics and can significantly enhance the online behavior of learned parameterizations.
} \vspace{-0.2cm}
\label{fig:online_energy_performance}
\end{figure}

Following \cite{ross2023benchmarking}, we assess the online performance of the different models using \textit{spectral energy flux diagnostics}, which provide a summary of the scale-by-scale energetics of the simulated flow. We again focus on four physically meaningful quantities: the spectral flux of kinetic energy (\textbf{KEflux}), the spectral flux of available potential energy (\textbf{APEflux}), the spectral spectrum of APE generation rate (\textbf{APEgenspec}), and the spectrum of kinetic energy dissipation due to bottom drag (\textbf{KEfrictionspec}). These diagnostics collectively characterize the fundamental processes of energy transfer, generation, and dissipation in quasi-geostrophic turbulence and are therefore a stringent test of the physical fidelity of each modeling approach.

Figure~\ref{fig:online_energy_performance} extends Figure~\ref{fig:pyqg} by adding KSD-calibrated variants of the low-resolution and FCNN-augmented models; the high-resolution, low-resolution, and Low-res-FCNN spectra (blue, orange, and green) are identical to those shown in Figure~\ref{fig:pyqg}. As seen there, the FCNN-augmented model remains close to the baseline low-resolution curves and noticeably below the high-resolution reference, so the data-driven parameterization alone only modestly reduces the bias in the long-term statistics. In contrast, applying KSD calibration to the low-resolution model (red solid line) shifts the spectra toward the high-resolution target in all four diagnostics: for KEflux, APEflux, and APEgenspec, the KSD-calibrated curves closely track the amplitude and overall shape of the high-resolution spectra, and even for KEfrictionspec the calibrated spectrum moves in the correct direction relative to the uncalibrated low-resolution models. The combined Low-res-FCNN-KSD model (purple solid line) yields the best overall agreement with the high-resolution reference. These results show that KSD calibration can substantially repair the long-term spectral and statistical structure of the flow, and that it complements, rather than replaces, standard data-driven parameterizations trained on short-horizon prediction losses.

To quantify performance improvements, Figure~\ref{fig:online_similarity_scores} displays {\em improvement} scores calculated from spectral diagnostics \cite{ross2023benchmarking}. First, a {\em spectral diff} score is defined as the root mean squared difference between the spectra of each model and the reference high-resolution.

\begin{figure}[htbp!]
\centering
\includegraphics[width=0.81\textwidth]{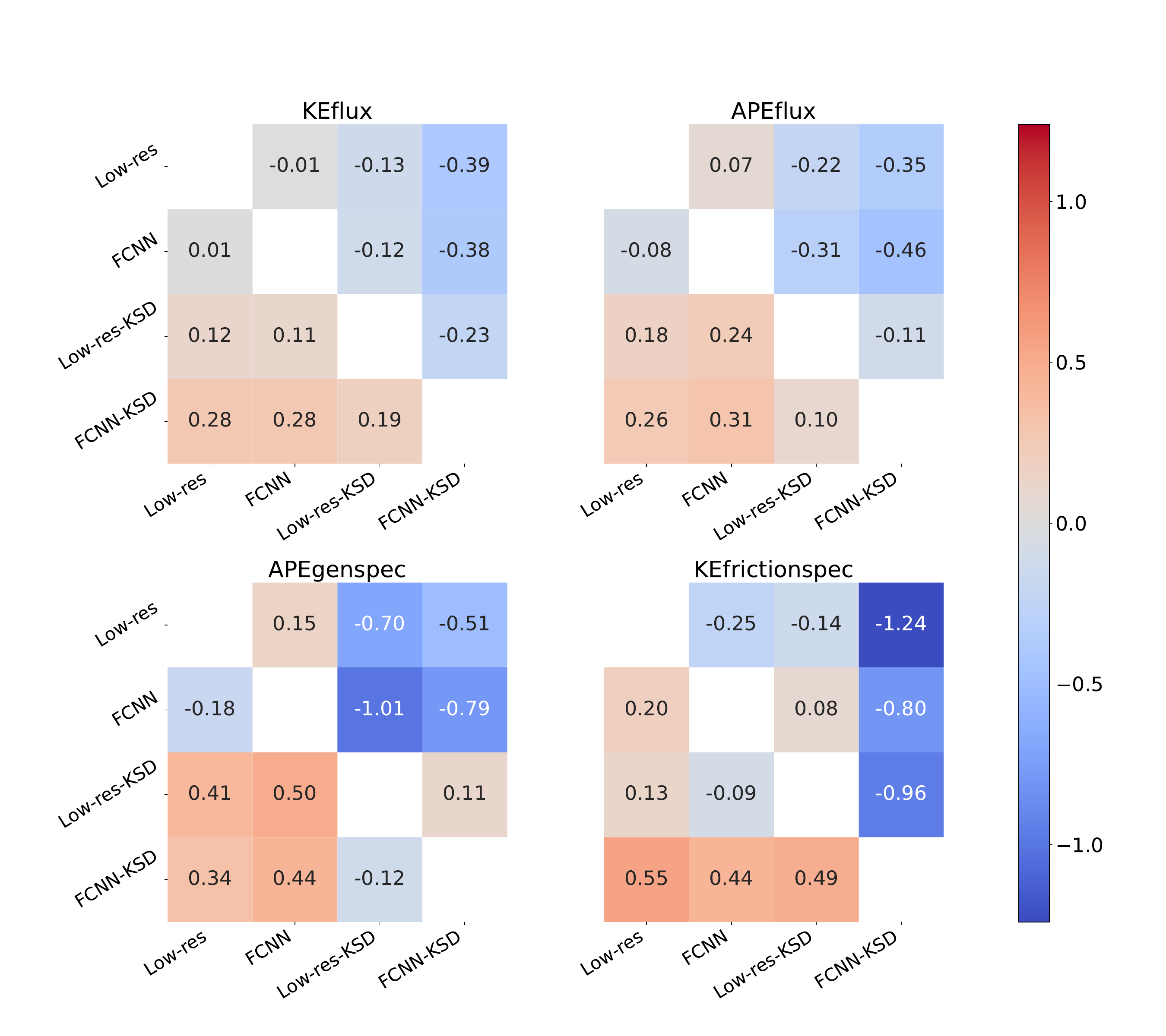}
\caption{Pairwise improvement scores between different model configurations, using the high-resolution simulation as the common target. For each spectral diagnostic (KEflux, APEflux, APEgenspec, KEfrictionspec), we form a matrix whose rows correspond to the \emph{model} and columns to the \emph{baseline}, and whose entries are the improvement scores defined in equation~\eqref{eq:similarity_score}. A positive value in cell $(\text{row},\text{column})$ means that the row model is closer to the high-resolution spectra than the column baseline for that diagnostic, whereas negative values indicate degradation. Warm colors (red) correspond to larger positive improvements and cool colors (blue) to larger negative values. Across all four diagnostics, entries in the rows associated with KSD-calibrated models (Low-res-KSD and FCNN-KSD) are predominantly positive when compared against the Low-res and FCNN baselines, while entries against KSD-based baselines are frequently negative for non-KSD models. Together with Figures~\ref{fig:pyqg} and \ref{fig:online_energy_performance}, these matrices show that KSD calibration consistently moves both the unparameterized low-resolution model and the FCNN-augmented model closer to the high-resolution target, with the combined FCNN-KSD configuration achieving the strongest overall agreement.
}
\label{fig:online_similarity_scores}
\end{figure}

The {\em improvement} score, evaluating the improvement of a {\em model} over a {\em baseline}, with respect to a {\em target}, is computed as:
\begin{equation}
\text{score} = 1 - \frac{\text{spectral\_diff(model, target)}}{\text{spectral\_diff(baseline, target)}}.
\label{eq:similarity_score}
\end{equation}
Higher values (closer to 1) indicate more substantial improvement over the baseline, while negative values imply degradation.

Figure~\ref{fig:online_similarity_scores} summarizes these effects quantitatively. Each panel shows the improvement scores for one spectral diagnostic, with rows corresponding to the \emph{model} and columns to the \emph{baseline} in equation~\eqref{eq:similarity_score}. Entries in the upper-left 2-by-2 block indicate that the FCNN provides at best modest and sometimes negative changes relative to the low-resolution model, consistent with the spectra in Figures~\ref{fig:pyqg}, \ref{fig:online_energy_performance}. In contrast, rows corresponding to the KSD-calibrated models (Low-res-KSD and FCNN-KSD) are predominantly positive when compared against either Low-res or FCNN across all four diagnostics, indicating that KSD systematically moves the models closer to the high-resolution target. Among all configurations, the combined FCNN-KSD model typically attains the highest improvement scores, reinforcing that KSD calibration enhances both unparameterized (Low-res) and learned (FCNN-augmented) models and provides a consistent, quantitative gain in the long-term spectral statistics.

\subsection{Computability}

KSD calibration is very computationally efficient. It does not require substantial training and tuning. All experiments were conducted on an Intel Core i7-14700KF (3.40GHz, 32GB RAM) and an NVIDIA GeForce RTX 4080 GPU (16GB).

\section{Conclusion and Discussion}
The central message of this paper is that, for chaotic systems where errors accumulate over time or correlate over space, calibrating systematic statistical outputs using limited physical knowledge or first-principles physics offers a cost-effective and stable alternative to employing computationally expensive hybrid loss functions in machine learning models. The proposed calibration framework is both efficient and capable of achieving comparable—or even superior—performance, as demonstrated here. This approach holds significant promise for climate system modeling, where long-term state trajectories are inherently unpredictable, yet statistical properties remain reliably predictable—even in coarse-resolution models. 
To the best of our knowledge, this work is the first to explicitly formalize and demonstrate this concept through both differentiable and non-differentiable programming in both offline and online assimilation settings.
A limitation of our approach is the assumption of access to the knowledge distribution, which may not always be readily available in practice. In future work, we aim to integrate our method within a distribution learning framework (e.g., \cite{Li2025GRL}) to jointly infer and calibrate against the steady-state distribution of a dynamical system. Exploring these avenues could open up new possibilities for improving predictions in more complex systems. 

\section*{Acknowledgements}
We acknowledge funding from NSF through the Learning the Earth with Artificial intelligence and Physics (LEAP) Science and Technology Center (STC) (Award \#2019625), and from the Columbia University Data Science Institute.

\section*{Data Availability}


All data associated with this manuscript, along with code necessary to process the data and generate the figures can be found at https://github.com/zwhou99/Distribution-Informed-Prediction.

\section*{Author Contributions}
The specific contributions of each author should be noted with each author denoted by their initials and starting on a new line.

ZH: Conceptualization, methodology development, design of experiments, data processing and archiving, data analysis and visualization, result interpretation, writing.

JS: Theoretical Justification.

SV: Conceptualization, PyQG experiment setup, data analysis.

PJ: Conceptualization, data curating and investigation.

SL: Conceptualization, data analysis, investigation and writing.

TZ: Conceptualization, methodology development, design of experiments, theoretical justification, data analysis and visualization, result interpretation, writing.


\clearpage
\newpage
\appendix
\setcounter{table}{0}
\renewcommand{\thetable}{A\arabic{table}}
\setcounter{figure}{0}
\renewcommand{\thefigure}{A\arabic{figure}}


\section{Theoretical justification}\label{theoryappendix}

Suppose we observe $\{ \hat y_i: i=1,\ldots,n\} $ from $ \hat y_i = y_i + \epsilon_i $, where $y_i $ denote unobserved true values, and $\epsilon_i$'s are i.i.d.\ $N(0,\sigma^2) $. Suppose $U(\hat y_{1:n})$ is an $n$-variable function of $\hat y_1,\ldots, \hat y_n$. Let 
$$
\tilde y_i = \hat y_i - \lambda \frac{\partial U(\hat y_{1:n})}{\partial \hat y_i}.
$$
Denote by $\mathbf{Y}, \mathbf{\hat Y}, \mathbf{\tilde Y} \in \mathbb{R}^n $ their vector version.
Then we have the following result.
\begin{proposition} \label{prop1}
For any twice differentiable function $U$, the choice of $\lambda =\frac{ \sigma^2\sum_{i=1}^n\mathbb{E}\left\{\frac{\partial^2 U(\hat y_{1:n})}{\partial \hat y_i^2} \right\} }{ \sum_{i=1}^n \mathbb{E}\left\{\frac{\partial U(\hat y_{1:n})}{\partial \hat y_i} \right\}^2 }$ minimizes  $ \mathbb{E}\|\mathbf{\tilde Y}-\mathbf{ Y}\|^2$. In particular, with this choice of $\lambda$, we have $ \mathbb{E}\|\mathbf{\tilde Y}-\mathbf{Y}\|^2 \leq \mathbb{E}\|\mathbf{\hat Y}-\mathbf{Y}\|^2 $.
The equality holds iff $\sum_{i=1}^n\mathbb{E}\left\{\frac{\partial^2 U(\hat y_{1:n})}{\partial \hat y_i^2} \right\}=0$.
\end{proposition}
\begin{proof}
    We apply the following celebrated Stein's Lemma \citep{csiszar2011information}:
    \begin{lemma}[Stein's Lemma]
    Let $Z \sim N(0,1)$. Let $h : \mathbb{R}\to  \mathbb{R}$ be an absolutely continuous function (differentiable is suﬀicient) such that $\mathbb{E}[|h'(Z)|]<\infty $. Then $\mathbb{E}[h'(Z)] = \mathbb{E}[Zh(Z)] $.
    \end{lemma}
    By Stein's Lemma, we have
    $$
    \mathbb{E}\left[ (\hat y_i - y_i) \left\{\frac{\partial U(\hat y_{1:n})}{\partial \hat y_i} \right\} \right] = \sigma^2 \mathbb{E}\left\{\frac{\partial^2 U(\hat y_{1:n})}{\partial \hat y_i^2} \right\}.
    $$
Hence we have
    \begin{align*}
    &~\mathbb{E}\|\mathbf{\tilde{Y}}-\mathbf{Y}\|^2 - \mathbb{E}\|\mathbf{\hat{Y}}-\mathbf{Y}\|^2\\ 
    =&~ \sum_{i=1}^n \mathbb{E}(\tilde y_i-y_i)^2 - \sum_{i=1}^n \mathbb{E}( \hat{y}_i-y_i)^2 \\  
    =&~ \sum_{i=1}^n\mathbb{E}\left\{\hat{y}_i-y_i -\lambda \frac{\partial U(\hat{y}_{1:n})}{\partial \hat{y}_i}\right\}^2 - \sum_{i=1}^n \mathbb{E}( \hat{y}_i-y_i)^2 \\
    =&~  - 2\lambda\sum_{i=1}^n\mathbb{E}\left[ (\hat{y}_i - y_i) \left\{\frac{\partial U(\hat{y}_{1:n})}{\partial \hat{y}_i} \right\} \right] + \lambda^2\sum_{i=1}^n \mathbb{E}\left\{\frac{\partial U(\hat{y}_{1:n})}{\partial \hat{y}_i} \right\}^2 \\
    =&~ - 2\lambda\sigma^2\sum_{i=1}^n\mathbb{E}\left\{\frac{\partial^2 U(\hat{y}_{1:n})}{\partial \hat{y}_i^2} \right\} + \lambda^2\sum_{i=1}^n \mathbb{E}\left\{\frac{\partial U(\hat{y}_{1:n})}{\partial \hat{y}_i} \right\}^2, 
    \end{align*}
    which is minimized with the choice of $\lambda =\frac{ \sigma^2\sum_{i=1}^n\mathbb{E}\left\{\frac{\partial^2 U(\hat{y}_{1:n})}{\partial \hat{y}_i^2} \right\} }{ \sum_{i=1}^n \mathbb{E}\left\{\frac{\partial U(\hat{y}_{1:n})}{\partial \hat{y}_i} \right\}^2 }$. 
    Therefore, with that choice of $\lambda$, $\mathbb{E}\|\mathbf{\tilde{Y}}-\mathbf{Y}\|^2$ is also minimized, and the minimal is
\begin{align*}
    \mathbb{E}\|\mathbf{\tilde{Y}}-\mathbf{Y}\|^2 
    =\mathbb{E}\|\mathbf{\hat{Y}}-\mathbf{Y}\|^2 - \frac{ \sigma^4\left[\sum_{i=1}^n\mathbb{E}\left\{\frac{\partial^2 U(\hat{y}_{1:n})}{\partial \hat{y}_i^2} \right\}\right]^2 }{ \sum_{i=1}^n \mathbb{E}\left\{\frac{\partial U(\hat{y}_{1:n})}{\partial \hat{y}_i} \right\}^2 }
    =n\sigma^2 - \frac{ \sigma^4\left[\sum_{i=1}^n\mathbb{E}\left\{\frac{\partial^2 U(\hat{y}_{1:n})}{\partial \hat{y}_i^2} \right\}\right]^2 }{ \sum_{i=1}^n \mathbb{E}\left\{\frac{\partial U(\hat{y}_{1:n})}{\partial \hat{y}_i} \right\}^2 }.
\end{align*}

    \hfill
\end{proof}

\newpage

\section{Implementation Details of the Application to air-sea CO2 flux}\label{pco2appendix}
In this appendix section, we provide more details for the application to air-sea CO2 flux example.

\subsection{Dataset information}

\paragraph{Community Earth System Model} The Community Earth System Model (CESM \cite{kay2015community}) is a sophisticated climate modeling framework developed by the National Center for Atmospheric Research (NCAR). CESM is designed to simulate the interactions between the atmosphere, oceans, land surface, and sea ice, providing a comprehensive tool for studying Earth's climate system. It integrates various components such as the atmosphere model (CAM), ocean model (POP), sea ice model (CICE), and land model (CLM), among others. CESM allows researchers to conduct experiments to understand climate variability and change, forecast future climate scenarios, and investigate the potential impacts of different factors on the global climate. Its modular structure enables flexibility and customization, making it a vital resource for climate scientists worldwide.

\paragraph{Surface Ocean CO2 Atlas} The Surface Ocean CO2 Atlas (SOCAT \cite{socat1,socat2}) is an extensive, collaborative effort to compile and provide access to quality-controlled observations of carbon dioxide (CO2) concentrations in the surface ocean. SOCAT aggregates data from numerous international research programs, offering a comprehensive dataset that spans multiple decades and covers all major ocean basins. This invaluable resource supports research into the global carbon cycle, ocean acidification, and the role of oceans in mitigating climate change. By standardizing and verifying CO2 measurements, SOCAT enables accurate assessments of spatial and temporal trends in oceanic CO2 levels, facilitating better understanding and modeling of the Earth's climate system and informing policy decisions related to carbon management and climate mitigation strategies.

\subsection{Replication of SOM-FFN}

For our example, we reconstructed the SOM-FFN on CESM001 by closely following the steps outlined by \cite{gloege2021quantifying} with minor adjustments. SOM-FFN is a non-linear regression method that integrates self-organizing maps (SOM) and feed-forward neural networks (FFN) to extrapolate sparse pCO\textsubscript{2} observations onto a global 1° × 1° grid at monthly intervals. The datasets we used include Sea Surface Temperature (SST) and Surface Chlorophyll-a (Chl-a) from satellite data, Sea Surface Salinity (SSS) from in-situ sources, Mixed Layer Depth (MLD) climatology from Argo floats, and atmospheric CO\textsubscript{2} mixing ratio (xCO\textsubscript{2}). These variables are essential proxies for processes affecting pCO\textsubscript{2}.

The initial step employed SOM to classify the global ocean into 16 biogeochemical provinces based on climatological variables, including surface ocean pCO\textsubscript{2} from \cite{takahashi2009climatological}, SST, SSS, and MLD. This classification leverages regional consistency in the main drivers of pCO\textsubscript{2} variability. Subsequently, a non-linear regression model is developed to estimate pCO\textsubscript{2} using the environmental drivers (SST, SSS, MLD, Chl-a, and xCO\textsubscript{2}). These variables vary monthly from 1982 to 2016, with any gaps filled by climatology or omitted. For each province, a separate FFN was developed based on Surface Ocean CO\textsubscript{2} ATlas (SOCAT) observations \citep{socat1, socat2} in the first 28 years (1982-2009), avoiding mechanistic assumptions. After training, pCO\textsubscript{2} in the last 7 years (2010-2016) was reconstructed (point-wise predictions) by FFN separately in each province.

\begin{figure}[htbp!]
\centering
\begin{tabular}{ll}
(A) Coverage of the 6th province.  & (B) Observation coverage in the 6th province. \\
\includegraphics[width=0.43\textwidth]{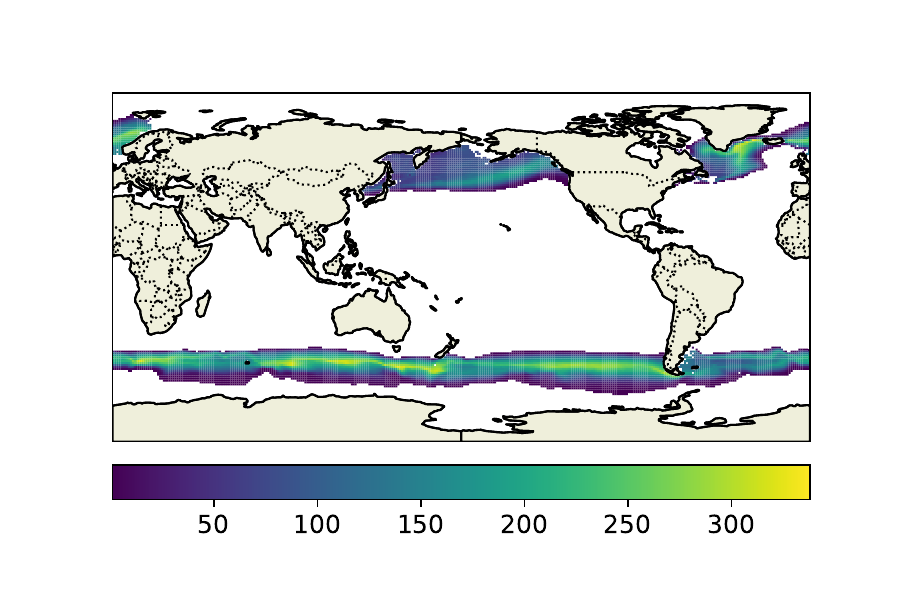} &  
\includegraphics[width=0.43\textwidth]{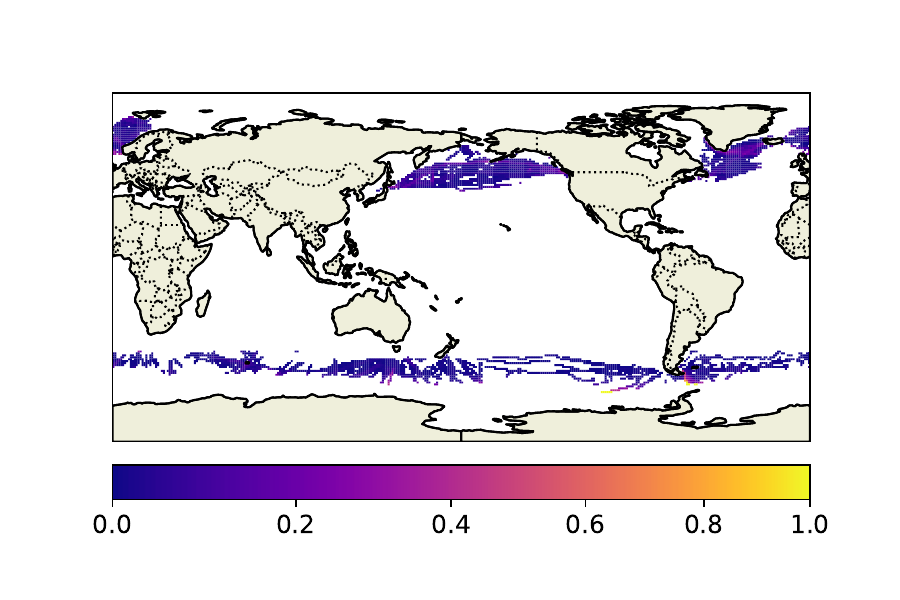}\\
(C) Data availability (all locations). & (D) Data availability (locations with observation).\\
\includegraphics[width=0.43\textwidth]{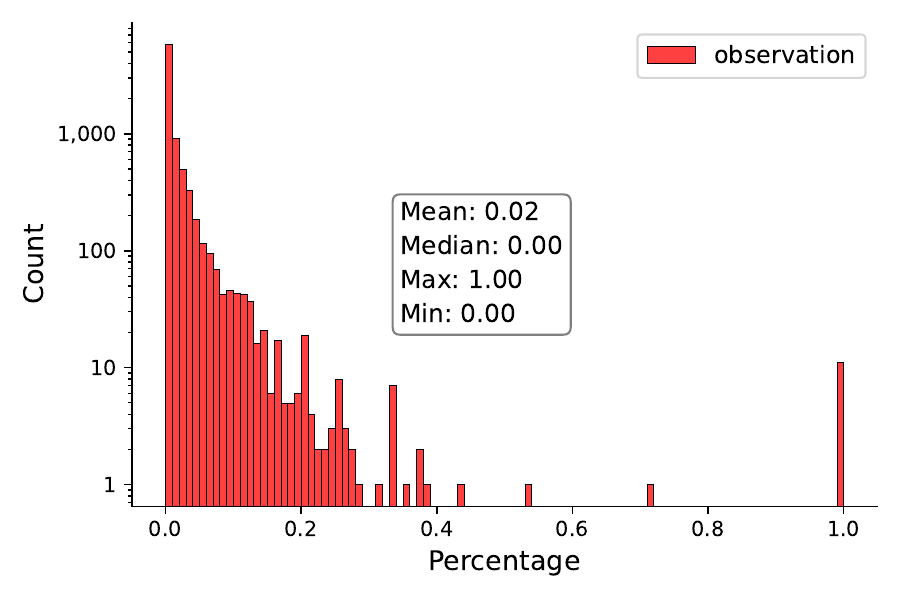} &
\includegraphics[width=0.43\textwidth]{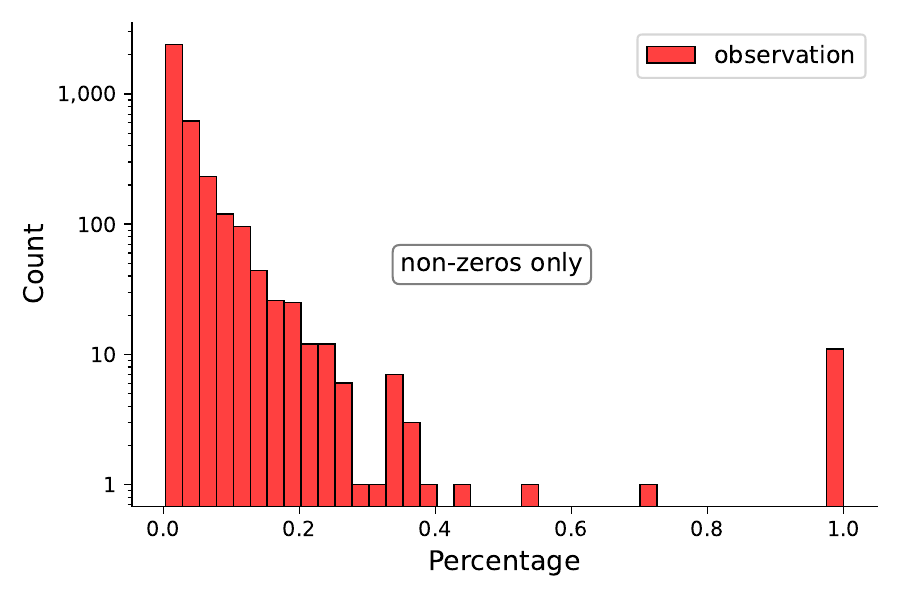}
\end{tabular}
    \caption{Heterogeneous observation coverage in the 6th province for modeling sea--air CO$_2$ fluxes.
(A) Frequency (in months) with which each grid cell is assigned to the 6th province according to the Self-Organizing Map (SOM) clustering from \cite{landschutzer2016decadal}, showing that only a subset of locations are persistently classified into this province while many others enter only occasionally. 
(B) Data availability percentage for months in which each grid cell belongs to the 6th province, highlighting strong spatial variability in observational coverage within the province.
(C) Histogram of data availability for all grid cells in the 6th province, including locations with zero observations. Most locations have essentially no data, and the mean availability is only about 2\%, indicating that observations are extremely sparse overall.
(D) Histogram of data availability restricted to locations with at least one observation. Among these observed locations, data availability spans a wide range from very sparse to near-complete coverage, revealing pronounced heterogeneity even within the subset of sampled grid cells.
Together, these panels emphasize that the 6th province is characterized by highly uneven and spatially structured data availability, which poses a challenge for learning and evaluating data--driven models.}
    
    \label{fig:obv_pro15_com}
\end{figure}

Among the 16 biogeochemical provinces, the 6th province exhibits highly heterogeneous spatial and temporal sampling. The \emph{denominator} map (Figure~\ref{fig:obv_pro15_com}-(A)) shows how often each grid cell is assigned to the 6th province over the full time series: a small set of locations are included frequently, whereas large portions of the province belong to it only intermittently or not at all. Conditional on a grid cell belonging to the 6th province, the \emph{percentage} map (Figure~\ref{fig:obv_pro15_com}-(B)) reports the percentage of those months in which SOCAT measurements are available, revealing strong spatial variability in observational coverage within the province.

Figure~\ref{fig:obv_pro15_com}-(C) and (D) summarize this heterogeneity in terms of data availability. Figure~\ref{fig:obv_pro15_com}-(C) shows the distribution of data availability across all province--grid cells: most locations have zero observations, the median availability is 0\%, and the mean availability is only about 2\%, indicating that observational coverage is extremely sparse overall. Figure~\ref{fig:obv_pro15_com}-(D) restricts attention to locations with at least one observation and shows that, even among sampled grid cells, data availability ranges from very rare to near-complete coverage.

On top of the already sparse and biased global SOCAT coverage \cite{socat1, socat2}, this pronounced and spatially structured heterogeneity hampers the FFN's ability to learn the full range of regional variability in this province and leads to substantial prediction errors, particularly in poorly observed areas. As a result, the calibration step becomes especially important here, as it leverages prior information on the marginal pCO\textsubscript{2} distribution to partially compensate for these observational gaps and improve the reconstruction in poorly observed areas.

\clearpage

\section{Implementation Details of the Application to an Online Hybrid Emulator of Quasi-Geostrophic Turbulence}\label{pyqgappendix}

In this appendix section, we provide additional results for the application to an Online Hybrid Emulator of Quasi-Geostrophic Turbulence.

\subsection{Model Setup}\label{sec:QGModelsetup}

We adopt a two-layer quasigeostrophic (QG) model to simulate large-scale baroclinic turbulence in a stratified, rotating fluid. The governing dynamics evolve the potential vorticities \( q_1 \) and \( q_2 \) in each layer according to
\begin{align}
\partial_t q_1 + J(\psi_1, q_1) + \beta \partial_x \psi_1 &= \text{ssd}, \\
\partial_t q_2 + J(\psi_2, q_2) + \beta \partial_x \psi_2 &= -r_{\text{ek}} \nabla^2 \psi_2 + \text{ssd},
\end{align}
where \( \psi_1 \) and \( \psi_2 \) are the streamfunctions, \( \beta \) is the planetary vorticity gradient, and \( r_{\text{ek}} \) is a linear bottom drag applied only to the lower layer. The right-hand side term \texttt{ssd} represents small-scale dissipation that removes enstrophy at high wavenumbers. The Jacobian operator is defined as \( J(A, B) = A_x B_y - A_y B_x \).

Potential vorticity in each layer is related to the streamfunction via
\begin{align}
q_1 &= \nabla^2 \psi_1 + F_1(\psi_2 - \psi_1), \\
q_2 &= \nabla^2 \psi_2 + F_2(\psi_1 - \psi_2),
\end{align}
with coupling coefficients \( F_1 = k_d^2 / (1 + \delta^2) \), and \( F_2 = \delta F_1 \). The deformation wavenumber \( k_d \) is defined as
\[
k_d^2 = \frac{f_0^2}{g} \cdot \frac{H_1 + H_2}{H_1 H_2},
\]
where \( H_1 \) and \( H_2 \) are the upper and lower layer depths, respectively, and \( H = H_1 + H_2 \) is the total depth of the fluid column.

All simulations are carried out in a square periodic domain of size \( L = 10^6 \, \mathrm{m} \). We use a reference setup termed the \emph{eddy configuration}, with layer depths \( H_1 = 500 \, \mathrm{m} \) and \( H_2 = 2000 \, \mathrm{m} \), reduced gravity \( g' = 9.81 \, \mathrm{m/s^2} \), and a zonal velocity shear \( \Delta U = 0.025 \, \mathrm{m/s} \). The Coriolis gradient is set to \( \beta = 1.5 \times 10^{-11} \, \mathrm{m}^{-1}\mathrm{s}^{-1} \), and the bottom drag coefficient is \( r_{\text{ek}} = 5.787 \times 10^{-7} \, \mathrm{s}^{-1} \). These choices give a deformation radius of \( r_d = 15{,}000 \, \mathrm{m} \).

To investigate the effect of resolution, we use two grid setups. The high-resolution case uses a \( 256 \times 256 \) grid, giving \( \Delta x \approx 3906 \, \mathrm{m} \) and a deformation-to-grid ratio \( r_d / \Delta x \approx 3.84 \), sufficient to resolve mesoscale dynamics. In the low-resolution case, we use a \( 64 \times 64 \) grid with \( \Delta x = 15{,}625 \, \mathrm{m} \), so that \( r_d / \Delta x \approx 0.96 \), falling below the eddy-resolving threshold.

Each simulation is initialized using the high-resolution model, which is integrated for 31,000 hours to allow the system to reach a statistically stationary state. Following this warm-up phase, the simulation proceeds for an additional 86,000 hours, during which a snapshot is recorded every 1,000 hours—resulting in a total of 87 snapshots spanning approximately 10 years of model time. The numerical time step is fixed at \( \Delta t = 1 \, \mathrm{hr} \).

\subsection{FCNN Training}

We replicate the training procedure of \citep{ross2023benchmarking}, using the public dataset released with their study.

Each FCNN is trained to predict one of the five subgrid forcing targets defined in \citep{ross2023benchmarking}. Inputs are derived from coarse-grained state variables, using all non-empty subsets of the three input fields described in the same study, yielding seven possible input combinations. Coupling each input combination with the five target types results in 35 distinct FCNN configurations.

Each model processes full $64 \times 64$ input and output fields, and is trained separately for each fluid layer. The architecture consists of eight convolutional layers with ReLU activations and circular padding. The first two layers use 128 and 64 filters, followed by six layers each with 32 filters. Batch normalization is applied to all layers except the output layer. Models are trained for 50 epochs using mean squared error (MSE) loss on mini-batches of size 64.


\subsection{Filtering and Coarse-Graining}

To generate coarse-resolution training data from high-resolution simulations, we apply a real-space filtering and averaging procedure. This approach is consistent with operations commonly used in general circulation models, where spectral representations are typically unavailable.

Specifically, we use the GCM-Filters package, which applies a diffusion-based filter designed to approximate the spectral transfer function of a Gaussian filter. The filter is implemented using polynomials of the Laplacian operator and applied directly on the high-resolution Cartesian grid.

After filtering, we perform spatial averaging to reduce the resolution by a factor of \( K \) in each horizontal direction. That is, we divide the domain into non-overlapping boxes of size \( K \times K \), and compute the average within each box to obtain coarse-grained fields. In all experiments, we set \( K = 4 \), corresponding to downsampling from \( 256 \times 256 \) to \( 64 \times 64 \).

This filtering and coarse-graining strategy avoids spectral assumptions and ensures compatibility with physical-space operations typical of ocean and climate models.

\newpage

\bibliographystyle{abbrvnat}  
\bibliography{KSD}

@inproceedings{chwialkowski2016kernel,
  title={A kernel test of goodness of fit},
  author={Chwialkowski, Kacper and Strathmann, Heiko and Gretton, Arthur},
  booktitle={International Conference on Machine Learning},
  pages={2606--2615},
  year={2016},
  organization={PMLR}
}

@inproceedings{liu2016kernelized,
  title={A kernelized Stein discrepancy for goodness-of-fit tests},
  author={Liu, Qiang and Lee, Jason and Jordan, Michael},
  booktitle={International Conference on Machine Learning},
  pages={276--284},
  year={2016},
  organization={PMLR}
}

@article{landschutzer2016decadal,
  title={Decadal variations and trends of the global ocean carbon sink},
  author={Landsch{\"u}tzer, Peter and Gruber, Nicolas and Bakker, Dorothee CE},
  journal={Global Biogeochemical Cycles},
  volume={30},
  number={10},
  pages={1396--1417},
  year={2016},
  publisher={Wiley Online Library}
}

@article{gloege2021quantifying,
  title={Quantifying errors in observationally based estimates of ocean carbon sink variability},
  author={Gloege, Lucas and McKinley, Galen A and Landsch{\"u}tzer, Peter and Fay, Amanda R and Fr{\"o}licher, Thomas L and Fyfe, John C and Ilyina, Tatiana and Jones, Steve and Lovenduski, Nicole S and Rodgers, Keith B and others},
  journal={Global Biogeochemical Cycles},
  volume={35},
  number={4},
  pages={e2020GB006788},
  year={2021},
  publisher={Wiley Online Library}
}

@article{kay2015community,
  title={The Community Earth System Model (CESM) large ensemble project: A community resource for studying climate change in the presence of internal climate variability},
  author={Kay, Jennifer E and Deser, Clara and Phillips, A and Mai, A and Hannay, Cecile and Strand, Gary and Arblaster, Julie Michelle and Bates, SC and Danabasoglu, Gokhan and Edwards, James and others},
  journal={Bulletin of the American Meteorological Society},
  volume={96},
  number={8},
  pages={1333--1349},
  year={2015}
}

@article{kovachki,
    author = "Kovachki, Nikola and Li, Zongyi and Liu, Burigede and Azizzadenesheli, Kamyar and Bhattacharya, Kaushik and Stuart, Andrew and Anandkumar, Anima",
    title = " Neural Operator: Learning Maps Between Function Spaces With Applications to PDEs",
    journal = "Journal of Machine Learning Research",
    year = {2023},
    volume = {24}
}

@article{ke2017lightgbm,
  title={Lightgbm: A highly efficient gradient boosting decision tree},
  author={Ke, Guolin and Meng, Qi and Finley, Thomas and Wang, Taifeng and Chen, Wei and Ma, Weidong and Ye, Qiwei and Liu, Tie-Yan},
  journal={Advances in Neural Information Processing Systems},
  volume={30},
  year={2017}
}

@article{schiff2024dyslim,
  title={DySLIM: Dynamics Stable Learning by Invariant Measure for Chaotic Systems},
  author={Schiff, Yair and Wan, Zhong Yi and Parker, Jeffrey B and Hoyer, Stephan and Kuleshov, Volodymyr and Sha, Fei and Zepeda-N{\'u}{\~n}ez, Leonardo},
  journal={arXiv preprint arXiv:2402.04467},
  year={2024}
}

@article{platt,
  title={Constraining chaos: Enforcing dynamical invariants in the training of recurrent neural networks},
  author={Platt, Jason A. and Penny, Stephen G. and Smith, Timothy A. and Chen, Tse-Chun and Abarbanel, Henry D. I.},
  journal={Chaos: An Interdisciplinary Journal of Nonlinear Science},
  volume={33},
  number={10},
  pages={103107},
  year={2023},
  publisher={AIP Publishing}
}

@article{ghil,
  title={The Physics of Climate Variability and Climate Change},
  author={Ghil, Michael},
  journal={International Journal of Weather, Climate and Climate Change},
  volume={10},
  number={3},
  pages={123-150},
  year={2002},
  publisher={American Meteorological Society},
}

@misc{pyqg,
  author = {Abernathey, Ryan and Rocha, Ceci{\'e}lio B. and Ross, Andrew and Jansen, Malte and Li, Zilu and Poulin, Francis J. and others},
  title = {pyqg/pyqg: V0.7.2 [Dataset]},
  year = {2022},
  publisher = {Zenodo}
}

@article{takahashi2009climatological,
  title={Climatological mean and decadal change in surface ocean pCO2, and net sea--air CO2 flux over the global oceans},
  author={Takahashi, Taro and Sutherland, Stewart C and Wanninkhof, Rik and Sweeney, Colm and Feely, Richard A and Chipman, David W and Hales, Burke and Friederich, Gernot and Chavez, Francisco and Sabine, Christopher and others},
  journal={Deep Sea Research Part II: Topical Studies in Oceanography},
  volume={56},
  number={8-10},
  pages={554--577},
  year={2009},
  publisher={Elsevier}
}

@article{socat1,
  title={Surface Ocean CO 2 Atlas (SOCAT) gridded data products},
  author={Sabine, Christopher L and Hankin, Steven and Koyuk, Heather and Bakker, Dorothee CE and Pfeil, Benjamin and Olsen, Are and Metzl, Nicolas and Kozyr, Alex and Fassbender, A and Manke, A and others},
  journal={Earth System Science Data},
  volume={5},
  number={1},
  pages={145--153},
  year={2013},
  publisher={Copernicus GmbH}
}

@article{socat2,
  title={A multi-decade record of high-quality fCO 2 data in version 3 of the Surface Ocean CO 2 Atlas (SOCAT)},
  author={Bakker, Dorothee CE and Pfeil, Benjamin and Landa, Camilla S and Metzl, Nicolas and O'brien, Kevin M and Olsen, Are and Smith, Karl and Cosca, Cathy and Harasawa, Sumiko and Jones, Stephen D and others},
  journal={Earth System Science Data},
  volume={8},
  number={2},
  pages={383--413},
  year={2016},
  publisher={Copernicus GmbH}
}

@article{cuomo2022scientific,
  title={Scientific machine learning through physics--informed neural networks: Where we are and what’s next},
  author={Cuomo, Salvatore and Di Cola, Vincenzo Schiano and Giampaolo, Fabio and Rozza, Gianluigi and Raissi, Maziar and Piccialli, Francesco},
  journal={Journal of Scientific Computing},
  volume={92},
  number={3},
  pages={88},
  year={2022},
  publisher={Springer}
}

@article{jospin2022hands,
  title={Hands-on Bayesian neural networks—A tutorial for deep learning users},
  author={Jospin, Laurent Valentin and Laga, Hamid and Boussaid, Farid and Buntine, Wray and Bennamoun, Mohammed},
  journal={IEEE Computational Intelligence Magazine},
  volume={17},
  number={2},
  pages={29--48},
  year={2022},
  publisher={IEEE}
}

@article{gottwald2016stochastic,
  title={Stochastic climate theory},
  author={Gottwald, Georg A and Crommelin, Daan T and Franzke, Christian LE},
  journal={arXiv preprint arXiv:1612.07474},
  year={2016}
}

@article{gorham2015measuring,
  title={Measuring sample quality with {S}tein's method},
  author={Gorham, Jackson and Mackey, Lester},
  journal={Advances in Neural Information Processing Systems},
  volume={28},
  year={2015}
}

@misc{gorham_measuring_2020,
	title = {Measuring {Sample} {Quality} with {Kernels}},
	url = {http://arxiv.org/abs/1703.01717},
	doi = {10.48550/arXiv.1703.01717},
	urldate = {2024-05-21},
	publisher = {arXiv},
	author = {Gorham, Jackson and Mackey, Lester},
	month = oct,
	year = {2020},
	note = {arXiv:1703.01717 [cs, stat]},
}

@article{liu2016stein,
  title={Stein variational gradient descent: A general purpose bayesian inference algorithm},
  author={Liu, Qiang and Wang, Dilin},
  journal={Advances in Neural Information Processing Systems},
  volume={29},
  year={2016}
}

@article{chattopadhyay_data-driven_2020,
	title = {Data-driven predictions of a multiscale {Lorenz} 96 chaotic system using machine-learning methods: reservoir computing, artificial neural network, and long short-term memory network},
	volume = {27},
	issn = {1023-5809},
	shorttitle = {Data-driven predictions of a multiscale {Lorenz} 96 chaotic system using machine-learning methods},
	url = {https://npg.copernicus.org/articles/27/373/2020/},
	doi = {10.5194/npg-27-373-2020},
	language = {English},
	number = {3},
	urldate = {2024-05-20},
	journal = {Nonlinear Processes in Geophysics},
	author = {Chattopadhyay, Ashesh and Hassanzadeh, Pedram and Subramanian, Devika},
	month = jul,
	year = {2020},
	note = {Publisher: Copernicus GmbH},
	pages = {373--389},
}

@article{franzke2020structure,
  title={The structure of climate variability across scales},
  author={Franzke, Christian LE and Barbosa, Susana and Blender, Richard and Fredriksen, Hege-Beate and Laepple, Thomas and Lambert, Fabrice and Nilsen, Tine and Rypdal, Kristoffer and Rypdal, Martin and Scotto, Manuel G and others},
  journal={Reviews of Geophysics},
  volume={58},
  number={2},
  pages={e2019RG000657},
  year={2020},
  publisher={Wiley Online Library}
}

@misc{gloege_2019, title={Large ensemble pCO2 testbed}, url={https://figshare.com/collections/Large_ensemble_pCO2_testbed/4568555/2}, DOI={10.6084/m9.figshare.c.4568555.v2}, abstractNote={This is a collection of randomly selected ensemble members from 4 large ensemble projects:<div><div>- <b>CanESM2</b> (http://data.ec.gc.ca/data/climate/scientificknowledge/the-eccc-climate-model-datasets-for-climate-science-and-impacts-research/the-canadian-earth-system-model-large-ensembles/)</div></div><div>- <b>CESM-LENS</b> (http://www.cesm.ucar.edu/projects/community-projects/LENS/)</div><div>- <b>GFDL </b>( http://poseidon.princeton.edu)</div><div>- <b>MPI-GE </b>(https://mpimet.mpg.de/en/grand-ensemble/)</div><div><br></div><div>Each ensemble member was interpolated from its native grid to a 1x1 degree lat/lon grid. The variables are monthly over the 1982-2017 time frame and sampled as the SOCATv5 data product. Historical atmospheric CO2 is used up to 2005 with RCP8.5 after 2005. </div><div><br></div><div>The intention of this dataset is to evaluate ocean pCO2 gap-filling techniques. </div><div><br></div><div>Ensemble members from other large ensemble projects will be added to this dataset in the future. In total , there are 100 unique ensemble members across four large ensembles. Gap-filling techniques can be evaluated across 100 unique climate states. </div>}, publisher={figshare}, author={Gloege, Lucas}, year={2019}, month={Dec} }

@book{csiszar2011information,
  title={Information Theory: Coding Theorems for Discrete Memoryless Systems},
  author={Csisz{\'a}r, Imre and K{\"o}rner, J{\'a}nos},
  year={2011},
  publisher={Cambridge University Press}
}

@article{frezat2022posteriori,
  title={A posteriori learning for quasi-geostrophic turbulence parametrization},
  author={Frezat, Hugo and Le Sommer, Julien and Fablet, Ronan and Balarac, Guillaume and Lguensat, Redouane},
  journal={Journal of Advances in Modeling Earth Systems},
  volume={14},
  number={11},
  pages={e2022MS003124},
  year={2022},
  publisher={Wiley Online Library}
}

@article{li2024machine,
  title={Machine-assisted physical closure for coarse suspended sediments in vegetated turbulent channel flows},
  author={Li, Shuolin and Qu, Yongquan and Zheng, Tian and Gentine, Pierre},
  journal={Geophysical Research Letters},
  volume={51},
  number={20},
  pages={e2024GL110475},
  year={2024},
  publisher={Wiley Online Library}
}

@article{ross2023benchmarking,
  title={Benchmarking of machine learning ocean subgrid parameterizations in an idealized model},
  author={Ross, Andrew and Li, Ziwei and Perezhogin, Pavel and Fernandez-Granda, Carlos and Zanna, Laure},
  journal={Journal of Advances in Modeling Earth Systems},
  volume={15},
  number={1},
  year={2023},
  publisher={Wiley Online Library}
}

@unpublished{Li2025GRL,
  author    = {Shuolin Li and Tian Zheng and Alban Farchi and Marc Bocquet and Pierre Gentine},
  title     = {Probabilistic Data Assimilation for Ensemble Distribution Projections via Generative Machine Learning},
  note      = {Manuscript submitted to \textit{Geophysical Research Letters}},
  year      = {2025},
}

\end{document}